\documentclass[12pt]{article}
\usepackage{url} 
\usepackage[english]{babel}
\usepackage[T2A]{fontenc}
\usepackage[utf8]{inputenc}
\usepackage{setspace}
\usepackage{graphicx}
\usepackage{float}
\usepackage{verbatim}
\usepackage{listings}
\usepackage{color}
\usepackage{subcaption}
\usepackage{tabularx}
\usepackage{natbib}
\usepackage{xcolor}
\usepackage{multicol}
\usepackage{textcomp}
\usepackage{amsthm}
\usepackage{thmtools, thm-restate}
\usepackage{stackengine}
\usepackage{numprint}
\usepackage[normalem]{ulem}
\stackMath
\usepackage[title]{appendix}
\usepackage{amsmath}
\usepackage{amssymb}
\usepackage{physics}
\usepackage{mathtools}
\usepackage{accents}
\usepackage{hyperref}
\usepackage{algorithm}
\usepackage{algpseudocode}
\usepackage{nameref,zref-xr}
\zxrsetup{toltxlabel}

\newcommand{\blind}{0}

\makeatletter
\newcommand*{\addFileDependency}[1]{
  \typeout{(#1)}
  \@addtofilelist{#1}
  \IfFileExists{#1}{}{\typeout{No file #1.}}
}
\makeatother

\newcommand{\defeq}{\vcentcolon=}

  \let\code=\texttt

\DeclareMathOperator*{\argmax}{argmax}

\newsavebox\MBox

\useunder{\uline}{\ul}{}

\newtheorem{remark}{Remark}

\graphicspath{{Figures/Vector/}{Figures/}}

\begin{document}

\def\spacingset#1{\renewcommand{\baselinestretch}%
{#1}\small\normalsize} \spacingset{1}


\if0\blind
{
  \title{\bf A Change-Point Approach to Estimating the Proportion of False Null Hypotheses in Multiple Testing}
  \author{Anica Kostic\thanks{Author for correspondence. [E-mail: a.kostic@lse.ac.uk, Address: Department of Statistics, London School of Economics and Political Science, Columbia House, Houghton Street, London, WC2A 2AE, UK]}\hspace{.2cm}
    and \\
    Piotr Fryzlewicz \\
    Department of Statistics, London School of Economics and Political Science}
  \date{}
  \maketitle
} \fi

\if1\blind
{
  \bigskip
  \bigskip
  \bigskip
  \begin{center}
    {\LARGE\bf Title}
\end{center}
  \medskip
} \fi

\bigskip
\begin{abstract}
For estimating the proportion of false null hypotheses in multiple testing, a family of estimators by \cite{storey2002direct} is widely used in the applied and statistical literature, with many methods suggested for selecting the parameter $\lambda$. Inspired by change-point concepts, our new approach to the selection of $\lambda$ first approximates the $p$-value plot with a piecewise linear function with a single change-point and then selects the $p$-value at the change-point location as $\lambda$.
We provide asymptotic theory for our estimator, relying on the theory of quantile processes. We propose a method that adapts to the unknown sparsity of the false-null $p$-values by tuning the parameter of the change-point estimation step. Additionally, we propose an application in the change-point literature and illustrate it using high-dimensional copy number variation (CNV) data.
\end{abstract}

\noindent%
\vfill

\newpage
\spacingset{2} 

\section{Introduction} \label{sec:intro}
 
Of interest to us in this work is the problem of estimating the proportion of false null hypotheses when many are tested simultaneously. Under the standard assumption of uniformity of true null $p$-values, $p$-value distribution is modeled as a mixture, with cumulative distribution function (CDF)
\begin{equation} \label{eq:mixture_model}
    F(x)=\pi_1F_1(x)+\pi_0x, \quad x\in [0,1],
\end{equation}
where $\pi_1$ is the unknown proportion of false null hypotheses, $\pi_0=1-\pi_1$ and $F_1$ is the CDF of the $p$-values under the alternative \citep{storey2002direct, Meinshausen2006, patra2016estimation}. 

Estimating the false null proportion $\pi_1$ comes down to estimating the proportion parameter of a two-component mixture distribution, with one component being the known distribution under the null. The proportion parameter quantifies the overall magnitude of significant deviations from baseline non-significant behavior, making it independently valuable. Specifically, this measure finds applications in astronomy and astrophysics \citep{Meinshausen2006, patra2016estimation, swanepoel1999limiting}. Furthermore, this measure is of interest in the classification literature, when only positive and unlabeled examples are available \citep{blanchard10a, Jain2016EstimatingTC, Jain2017}.
In the multiple testing literature, proportion estimators are mainly of indirect interest as they can be used to increase the power of the Benjamini-Hochberg based procedures \citep{benjamini1995controlling, benjamini2000adaptive}.

Storey's method \citep{storey2002direct}, initially introduced in \cite{schweder1982plots}, is the most common approach to the problem of estimating the proportion parameter.
Assuming that $F_1(x) \approx 1$ for sufficiently large $x\in (0,1)$, the CDF of the $p$-values is approximately linear with slope $\pi_0$. This linearity approximation gives rise to
Storey's family of plug-in proportion estimators, defined as follows:
\begin{equation} \label{eq:ss_pi1_est}
    \hat{\pi}_0(\lambda) = \frac{1- \hat{F}_n(\lambda)}{1-\lambda}, \quad \lambda \in (0,1),
\end{equation}
and $\hat{\pi}_1(\lambda) = 1-  \hat{\pi}_0(\lambda)$.
There are multiple estimators in the literature based on Storey's family, each proposing different tuning parameter value, with no general agreement on the optimal value of $\lambda$ \citep{benjamini2000adaptive, storey2003statistical, Storey2004, Jiang2008}. In general, a smaller $\lambda$ introduces higher bias, while choosing $\lambda$ close to 1 increases the variance of the proportion estimator. 


In this paper, we propose a new data-driven method for tuning Storey's estimator, which we call ``Difference of Slopes'' or DOS. We propose to approximate the plot of sorted $p$-values $(i,p_{(i)})$, $i = 1,\dots, n$, with a piecewise linear function with a single change in slope, using a statistic inspired by the change-point literature. If $1\le \hat{k}\le n$ is the estimated change-point location, we set $\lambda = p_{(\hat{k})}$ in Storey's estimator (\ref{eq:ss_pi1_est}) to obtain the proportion estimator, referred to as ``DOS-Storey''. This value of $\lambda$ aims to separate true from false null $p$-values. Specifically, we aim for it to be the smallest value at which $F_1(\lambda) \approx 1$, marking the onset of the linear part in the quantile function. By choosing such $\lambda$ our goal is to reduce the variance while maintaining low bias in the corresponding Storey's estimator. This method is introduced and described in more detail in Section \ref{sec:DOS_desc}. 
An illustration of the piecewise linear approximation produced by our method in a sparse and less sparse case is shown in Figure \ref{fig:pcws_lin_approx}.

\begin{figure}[H]
\begin{tabular}{cccc}
\subfloat{\includegraphics[width=.5\textwidth]{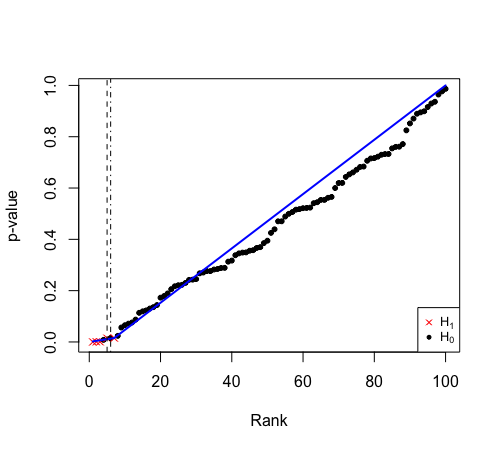}} &
\subfloat{\includegraphics[width=.5\textwidth] 
{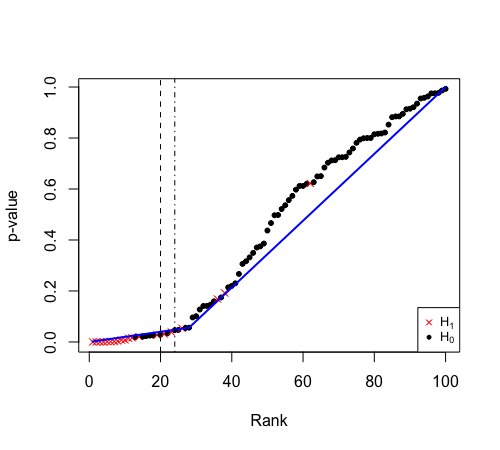}}
\end{tabular}
\caption{The plots illustrate our method in two settings. $p$-values are sorted and shown as dots for true null and crosses for false null hypotheses. The vertical dashed line is at $x = n_1$, the true number of false null hypotheses. The piecewise linear approximation is shown as the solid broken line and the vertical dot-dashed line is at  $x = \hat{n}_1$, the estimated number of false null hypotheses using our method. $p$-values are derived from Gaussian mean testing, $H_0: \mu = 0$ versus $H_1: \mu > 0$. The test statistics $T_i$ have $N(0,1)$ distribution under the null and $N(\mu_1,1)$ under the alternative, with $p$-values calculated as $p_i = 1 - \Phi(T_i)$. Left (sparse case): $n_1=5$, $\mu_1 = 3$, $\hat{n}_1 = 6$. Right (dense case): $n_1 = 20$, $\mu_1 = 2$, $\hat{n}_1 = 24$.}
\label{fig:pcws_lin_approx}
\end{figure}

The proposed change-point detection statistic can be tuned so that the resulting proportion estimator performs better in either highly sparse or moderately sparse cases. In Section \ref{sec:DOS_desc}, we also introduce an adaptive method that works well in both scenarios, eliminating the need for manually selecting the tuning parameter.

This paper is organized as follows. Section \ref{sec:DOS_desc} describes our proposed method. Theoretical results are presented in Section \ref{sec:theory}, while Section \ref{sec:simulations} contains the simulation study. The real data example can be found in Section \ref{sec:realdata}.
The Supplementary Material contains proofs and additional simulations and discussions.
The code implementing the introduced approach, the simulation study, and the real data example is included in the \code{R} package \code{MTCP}, available at \href{https://github.com/anicakostic/MTCP}{https://github.com/anicakostic/MTCP}.

\section{DOS Statistic and the DOS-Storey Estimator}\label{sec:DOS_desc}

Consider the sequence of sorted $p$-values, $p_{(1)}, \ldots, p_{(n)}$, and their representation as points $(i,p_{(i)})$ for $i = 1,\dots, n$, forming a $p$-value plot. The proposed piecewise linear approximation of the $p$-value plot is determined by the change-point location $\hat{k}$. It consists of a line connecting $(0,0)$ and $(\hat{k},p_{\hat{k}})$, and another line connecting $(\hat{k},p_{\hat{k}})$ and $(n+1,1)$.
To calculate the change-point estimate $\hat{k}$, we first define the Difference of Slopes (DOS) sequence as 
\begin{equation}
\label{eq:diff_slopes_seq}
    d_{\alpha}(i) =  \frac{p_{(2i)}-p_{(i)}}{i^\alpha}-\frac{p_{(i)}}{i^\alpha}  = \frac{p_{(2i)}-2p_{(i)}}{i^\alpha}
\end{equation}
for some $\alpha\in [1/2,1]$, $i = 1,\dots, \lfloor n/2 \rfloor$.  The DOS statistic, serving as the change-point estimator, is defined as the index of the maximum term in the DOS sequence:
\begin{equation} \label{eq:def_dos}
   \hat{k}_{\alpha} = \argmax_{nc_n \le i \le n/2} d_\alpha(i).
\end{equation}
The choice of the non-random sequence $c_n$ and the value of $\alpha$ is discussed below.
To obtain the proportion estimate using the DOS statistic, we plug $\lambda = p_{(\hat{k}_{\alpha})}$ into Storey's estimator (\ref{eq:ss_pi1_est}) and get the proposed DOS-Storey($\alpha$) false null proportion estimator:
\begin{equation} \label{eq:dos_prop_est}
    \hat{\pi}_1^{\alpha}= 
    \frac{\hat{k}_\alpha/n-p_{(\hat{k}_{\alpha})}}{1-p_{(\hat{k}_{\alpha})}}.
\end{equation}

To ensure the asymptotic results stated in Section \ref{sec:theory} hold, we exclude the initial $nc_n$ values of $d_\alpha(i)$ from the search for maximum in (\ref{eq:def_dos}). Precisely, the sufficient conditions are
$\frac{nc_n}{\log\log n} \to \infty,
    \frac{\log\log(1/c_n)}{\log\log n} \to C <\infty. $
In Remark 1 in the Supplementary Material, we discuss how different rates for $c_n$ affect the asymptotic results. The practical selection of $c_n$ is addressed in Section \ref{sec:simulations}.

An illustration of the proposed method with $\alpha = 1$ for the two examples from Figure \ref{fig:pcws_lin_approx} is given in Figure \ref{fig:illust_dos_sparse_dense}. The two plots in the right column show the dependency of $\hat{\pi}_1(\lambda)$ on $\lambda$, and the substantial influence of $\lambda$ due to the bias-variance trade-off, for two different values of $\pi_1$. In the sparse case (top), the estimated change-point location is at $\hat{k}_1=7$ and the estimated number of false nulls is $\hat{n}_1 = 6$ (true number is $n_1 = 5$).  In the dense case (bottom), the change-point location is at $\hat{k}_1=28$ with an estimated number of false nulls of $\hat{n}_1 = 24$ ($n_1 = 20$). Figure \ref{fig:illust_dos_sparse_dense} also shows how the false null $p$-values exceeding the threshold $p_{(\hat{k}_{\alpha})}$ are relatively few compared to the true null $p$-values. In both scenarios, it is evident how our approach effectively reduces variance and maintains low bias in the associated Storey's estimator (\ref{eq:dos_prop_est}). In larger samples, this effect is best seen when the proportion of false null hypotheses is small, as shown in the simulation study in Section \ref{sec:simulations}.

\begin{figure}[t]
\begin{tabular}{cccc}
\centering
\subfloat{\includegraphics[width=.45\textwidth]{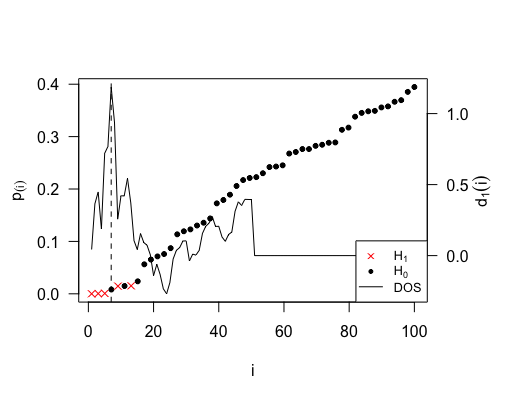}} &
\subfloat{\includegraphics[width=.45\textwidth]{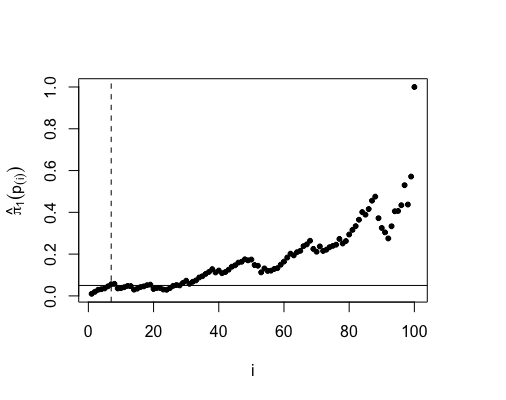}} \\
\subfloat{\includegraphics[width=.45\textwidth]{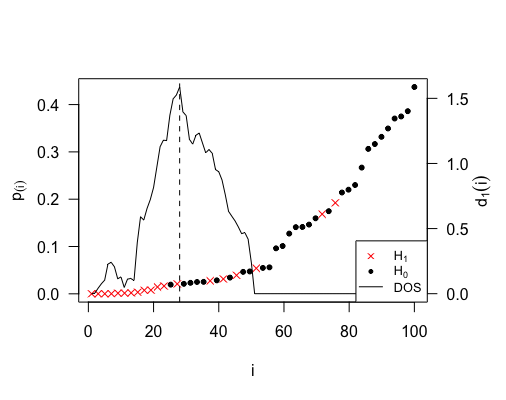}}&
\subfloat{\includegraphics[width=.45\textwidth]{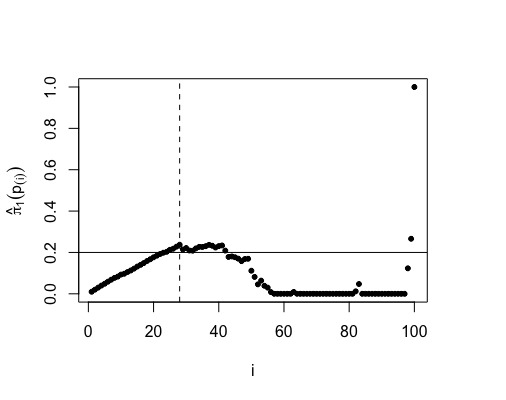}} 
\end{tabular}
\caption{Illustration of the DOS-Storey(1) method for the two examples from Figure \ref{fig:pcws_lin_approx}. Left column: $p$-value plot and the corresponding DOS sequence (solid line) with the estimated change-point location (vertical dashed line) in the sparse (top) and the dense (bottom) case. Right column: The sequence of Storey's estimators (\ref{eq:ss_pi1_est}), with $\lambda = p_{(i)}$, $i = 1,\dots, n$ in the sparse (top) and the dense (bottom) case; solid horizontal line is at the unknown false null proportion level, and dashed vertical line marks $\hat{\pi}_1^1$, the proportion estimated by the DOS-Storey(1) method.}
\label{fig:illust_dos_sparse_dense}
\end{figure}

We now provide further explanation of the DOS statistic. We begin by discussing the impact of parameter $\alpha$ by examining its boundary values, $\alpha = 1$ and $\alpha = 1/2$, which allow us to interpret our statistic within the context of change-point literature. Additionally, we propose an adaptive method that combines these two values, eliminating the need for manual selection.

For $\alpha=1$, the first term in (\ref{eq:diff_slopes_seq}) is the slope of the line connecting points $(i,p_{(i)})$ and $(2i,p_{(2i)})$, while the second term corresponds to the slope of the line connecting $(0,0)$ and $(i,p_{(i)})$. Therefore, $d_1(i)$ is the sequence of slopes differences in the $p$-value plot, and $\hat{k}_{1}$ is the location of the maximum slopes difference. Let $s_j=p_{(j)}-p_{(j-1)}$ be the sequence of spacings of $p$-values and $p_{(0)}=0$. The DOS sequence can be written as
$d_{1}(i) = \frac{1}{i}\sum_{j=1}^{i} s_j- \frac{1}{i}\sum_{j=i+1}^{2i} s_j$.
Thus, the DOS statistic finds the maximum difference of means on symmetric ($(0, i), (i, 2i)$) and increasing intervals ($i=1,\dots, n/2$) in the spacings sequence. A similar statistic, aiming to detect shifts in the piecewise constant mean of an ordered sample, has been studied in the nonparametric change-point literature \citep{Brodsky1993}. Therefore, the DOS statistic can be viewed as a technique for fitting a piecewise constant function to the sequence of spacings $s_i$. An illustration of the piecewise constant fit to the spacings sequence is provided in Figure \ref{fig:piecewise_fitting_illustration}.

\begin{figure}[ht]
\begin{tabular}{cccc}
\subfloat{\includegraphics[width=.5\textwidth]{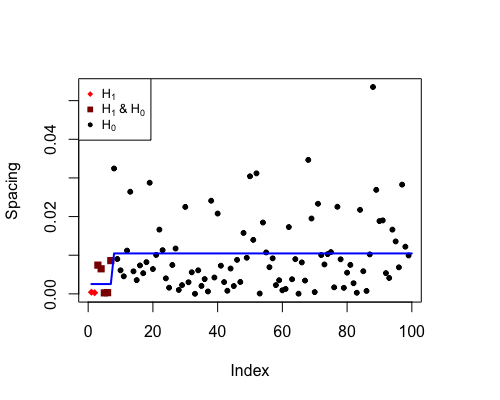}} &
\subfloat{\includegraphics[width=.5\textwidth]{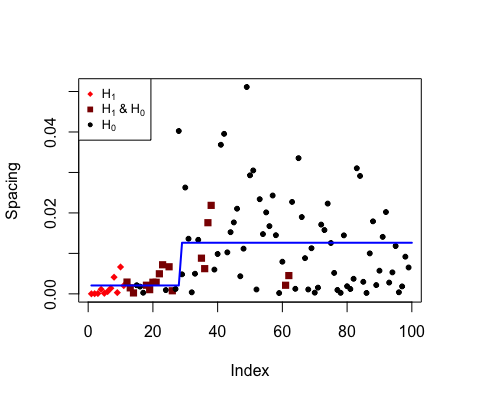}} 
\end{tabular}
\caption{The plots show the sequence of spacings $s_i$, and the piecewise constant fit to $s_i$ (solid line) using the DOS statistic with $\alpha=1$ for the two examples from Figure \ref{fig:pcws_lin_approx}. The spacings are denoted using different shapes: $s_i$ is represented with a diamond if both $p_{(i)}$ and $p_{(i-1)}$ are false null; with a circle if both are true null; and with a square if one $p$-value is true null, and the other is false null.}
\label{fig:piecewise_fitting_illustration}
\end{figure}

Similarly, for $\alpha=1/2$, we can interpret $d_{1/2}(i)$ as the scaled difference between the means of the first $i$ spacings ($s_1, \dots, s_i$) and the second $i$ spacings ($s_{i+1}, \dots, s_{2i}$).
In the context of the change-point literature, the statistic $\max_i d_{1/2}(i)$ can be regarded as a CUSUM-like statistic.

Theoretical results  of Section \ref{sec:theory} suggest that, asymptotically,  the estimated change-point location tends to occur earlier for larger $\alpha$, leading to a more conservative proportion estimate on average.
Simulation results of Section \ref{sec:simulations} show that using $\alpha = 1$ is more conservative and gives better results in very sparse cases, whereas $\alpha = 1/2$ performs better in case of many weaker false-null $p$-values. We propose a simple method of combining these two estimators, as shown in Algorithm \ref{alg:dos_adaptive}, yielding an adaptive DOS-Storey estimator (aDOS-Storey), denoted as $\hat{\pi}_1^a$, which performs well in both sparse and dense cases. 
The sparsity of the problem is first evaluated by thresholding the more conservative $\hat{\pi}_1^1$ estimator. Values of $\hat{\pi}_1^1$ below the threshold indicate that the problem is sparse and justify the use of $\alpha = 1$. Values above the threshold indicate less sparse cases, in which case we choose $\hat{\pi}_1^{1/2}$. We propose a specific threshold in Section \ref{sec:simulations} and find that the adaptive method successfully adjusts to the sparsity of the problem.

\begin{singlespacing}
    \begin{center}
\begin{algorithm}[H]
\caption{Adaptive DOS-Storey Method (aDOS-Storey)}
\hspace*{\algorithmicindent} \textbf{Input}: $p$-values sequence, threshold $\tau$\\
\hspace*{\algorithmicindent} \textbf{Output}: aDOS-Storey estimator $\hat{\pi}_1^a$
\begin{algorithmic}[1]
\If{\( \hat{\pi}_1^1 < \tau \)}
    \State $\hat{\pi}_1^a$ \( \gets \hat{\pi}_1^1 \)
\Else
    \State $\hat{\pi}_1^a$ \( \gets \hat{\pi}_1^{1/2} \)
\EndIf
\end{algorithmic}
\label{alg:dos_adaptive}
\end{algorithm}
\end{center}
\end{singlespacing}

The choice of interval pairs for comparing means is an important topic in the change-point literature. Typically, the largest difference in means on these intervals indicates a change-point. 
We use symmetric intervals ($(0, i), (i, 2i)$) to compute slopes, focusing on local behavior and changes in the quantile function. These intervals expand ($i=1,\dots, n/2$) to incorporate information from an increasing number of $p$-values until a shift to linearity is observed. As the estimated change-point is at most at location $\lfloor n/2 \rfloor$, our method is unsuitable for cases when the proportion of false null hypotheses is high.

We discuss a few alternative interval options. One alternative is to use the difference in slopes statistic on intervals $(1,i)$ and $(i,n)$ for $i = 1,\dots, n$, similar to the standard CUSUM statistic for estimating a single change-point. Simulation results in Section C.2 of the Supplementary Material indicate that this estimator performs reasonably well, but overall our estimator outperforms it. Another option is to use a sliding window of intervals $(i-m,i)$ and $(i, i+m)$ as in the MOSUM procedure \citep{eichingerkirch2018}. However, this approach risks detecting a change even when all $p$-values within the intervals are true null,  due to sample variability. Section D.2 of the Supplementary Material provides additional justification for using increasing in size intervals over sliding window intervals.

Using a change-point method for the purpose of tuning Storey's estimator has been previously mentioned in the literature. In \cite{benjamini2000adaptive}, it is implied that a change-point method can be used for estimating the proportion of false null hypotheses by identifying the end of the linear segment in the $p$-value plot. This approach is explored in \cite{Turkheimer2001} and \cite{Hwang2014}. However, the simulations included in Section C.2 of the Supplementary Material show that the latter two methods do not perform well when compared to our method.

Note that our scenario differs from the conventional change-point in the slope problem. The $p$-value plot has no typical change-point; instead, the change-point is linked to the suggested linear approximation. Asymptotically, the change-point location depends on the properties of the unknown quantile function and on our fitting procedure. Additionally, it is reasonable to think of our method as identifying a `knee' in the quantile plot, a topic briefly discussed in Section D.1 of the Supplementary Material.

\section{Theoretical Considerations}

\label{sec:theory}

We begin by introducing the assumptions on the $p$-value distribution given in  (\ref{eq:mixture_model}).
Different assumptions on $F_1$ can be found in the literature.  
Strong assumptions specify a family of distributions for $F_1$,  \citep{cai2007estimation, Pounds2003}, while weaker ones only restrict its shape. We place the following two assumptions on $F$: 
\begin{itemize}
    \item[(A1)] $F_1$ is a continuous distribution stochastically smaller than $U[0,1]$, with a weakly concave CDF. 
    \item[(A2)] Let 
            \begin{equation} \label{eq:ideal_func}
                 h_\alpha^F(t)\defeq \frac{F^{-1}(2t)-2F^{-1}(t)}{t^\alpha}, \quad t\in (0,1/2).
            \end{equation}
            $h_\alpha^F(t)$ has a unique point of local maximum at $\tilde{t}_\alpha\le 1/2$. \label{it:assumption}
\end{itemize}
Assumption (A1) implies that the density of the false null $p$-values is decreasing and is a common assumption in the literature \citep{langaas2005estimating, CELISSE20103132}.
Assumption (A2) is specific to our approach and is necessary for uniquely defining the asymptotic change-point location.  $h_\alpha^F$ represents the ``ideal function'' that the DOS sequence approximates, since ordered $p$-values are sample quantiles, $\frac{p_{(2i)}-2p_{(i)}}{(i/n)^\alpha} \approx h_\alpha^F(i/n)$. (A2) excludes cases when the signal is too weak, or when the proportion of non-null hypotheses is too large. Examples illustrating when (A2) does not hold are provided in Remark \ref{rem:assumption2} below.

\begin{restatable}[]{theorem}{nonunifmixconsist}
\label{thm:nonunif_mix_consist}
Consider the $p$-value distribution given in (\ref{eq:mixture_model}) and
assume that conditions (A1) and (A2) hold.
Let $p_{(1)},\dots,p_{(n)}$ be the order statistics of the iid sample from (\ref{eq:mixture_model}). Let $\hat{k}_{\alpha}$ and $\hat{\pi}_1^{\alpha}$ be as defined in (\ref{eq:def_dos}) and (\ref{eq:dos_prop_est}), respectively, with $c_n$ such that  $\frac{nc_n}{\log\log n} \to \infty$, 
    $\frac{\log\log(1/c_n)}{\log\log n} \to C <\infty$.   
It holds that
\begin{align}
    \hat{k}_{\alpha}/n &\stackrel{a.s.}{\to}\tilde{t}_\alpha \defeq\ \argmax_{0\le t\le 1/2} \frac{F^{-1}(2t)-2F^{-1}(t)}{t^\alpha} \label{eq:nonunif_consist_k},\\
    p_{(\hat{k}_{\alpha})}&\stackrel{a.s.}{\to} F^{-1}(\tilde{t}_\alpha) \label{eq:nonunif_consist_b},\\
    \hat{\pi}_1^{\alpha}&\stackrel{a.s.}{\to} \tilde{\pi}_1^\alpha \defeq \frac{\tilde{t}_\alpha-F^{-1}(\tilde{t}_\alpha)}{1-F^{-1}(\tilde{t}_\alpha)} \le \pi_1. \label{eq:nonunif_consist_pi1}
\end{align}
\end{restatable}

Theorem \ref{thm:nonunif_mix_consist} explains the asymptotic behavior of the estimated change-point location and the estimated proportion in terms of the ideal quantities, which are functionals of the $p$-value distribution quantile function. The convergence rates of the statistics in Theorem \ref{thm:nonunif_mix_consist} are considered within its proof in Section A of the Supplementary Material. These rates depend on the differentiability of $h_\alpha^F$ at $\tilde{t}_\alpha$, and the degree of ``flatness'' of $h_\alpha^F$ at $\tilde{t}_\alpha$, which is quantified using higher order derivatives.

When $\alpha_1<\alpha_2$, the decreasing function $1/t^{\alpha_2-\alpha_1}$ implies that $\argmax_t h_{\alpha_1}^F(t) > \argmax_t h_{\alpha_2}^F(t)$. This suggests that for larger $\alpha$ values, the ``change-point'' occurs later. From Theorem \ref{thm:nonunif_mix_consist} it also follows that the rate of convergence is slower for smaller $\alpha$.

The following two corollaries follow easily from Theorem \ref{thm:nonunif_mix_consist} and we provide them without proof. They illustrate the behavior of the proposed statistics in two specific cases.


\begin{restatable}[]{corollary}{corconsistunifmix}
\label{cor:unif_mix}
Let $p_{(1)},\dots, p_{(n)}$ be the order statistics of the iid sample of $p$-values coming from a mixture of two uniform distributions
\begin{equation} \label{eq:unif_mix}
    \pi_1U[0,b]+\pi_0U[0,1],
\end{equation}
where $0<b<1$.
Let $\hat{k}_{\alpha}$ and $\hat{\pi}_1^{\alpha}$, be the corresponding statistics proposed in (\ref{eq:def_dos}) and (\ref{eq:dos_prop_est}), respectively, with $c_n$ such that  $\frac{nc_n}{\log\log n} \to \infty$, 
    $\frac{\log\log(1/c_n)}{\log\log n} \to C <\infty$.  
It holds that $\hat{k}_{\alpha}/n \stackrel{a.s.}{\to}\pi_1+b\pi_0$,
     $p_{(\hat{k}_{\alpha})}\stackrel{a.s.}{\to} b$,
    $\hat{\pi}_1^{\alpha}\stackrel{a.s.}{\to} \pi_1$.
For large enough $n$, with probability one it holds that
$\abs{p_{(\hat{k}_{\alpha})}-b}\le C\frac{\log\log n}{nc_n^{2\alpha-1}}$;
$\abs{\hat{\pi}^{DOS}_1-\pi_1} \le C\frac{\log\log n}{nc_n^{2\alpha-1}}$.
Thus, $p_{(\hat{k}_{\alpha})}$ and $\hat{\pi}_1^{\alpha}$ are strongly consistent estimators of the uniform mixture parameters $b$ and $\pi_1$, respectively.
\end{restatable}

The following Corollary \ref{cor:non_unif_underest} shows that in general, when the support of the false null distribution is $[0,b]$, $p_{(\hat{k}_{\alpha})}$ will a.s. not overestimate $b$, so $\hat{\pi}_1^{\alpha}$ will be a conservative estimator of the proportion.

\begin{restatable}[]{corollary}{nonunifunderest}
\label{cor:non_unif_underest}
Let $[0,b]$, $b\le 1$ be the support of the alternative distribution $F_1$, where $F_1$ is stochastically smaller than $U[0,b]$ distribution, in the sense that $F_1(t)\ge t/b$ for all $0\le t\le b$. Then, $p_{(\hat{k}_{\alpha})}$ is an almost surely conservative estimator of the support boundary $b$, meaning that $p_{(\hat{k}_{\alpha})} \stackrel{a.s.}{\to} \tilde{b}$ where $\tilde{b}\le b$.
\end{restatable}

\begin{remark}
\label{rem:assumption2}
\normalfont    
There are two scenarios where Assumption 2 can fail. First, if $h_\alpha^F(t)$ is constant on an interval where it achieves its maximum value. This constant behavior of $h_\alpha^F(t)$ is a highly specific scenario and not easily characterized by conditions on $F$. It can happen when the $p$-value distribution is a mixture of uniforms. 
Second, it is possible for $h_\alpha^F(t)$ to be increasing on $[0,1/2]$ if the signal is too weak or the false null proportion is too large. These examples are shown in Figures \ref{fig:counterexample_unif} and \ref{fig:counterexample_gaus}. To illustrate that for a typical $p$-value model these scenarios do not pose an issue, we obtained numerical results in Mathematica, calculating $\tilde{t}_\alpha$, the asymptotic change-point location, for the one sided $p$-values from a two-point Gaussian mixture $(1-\pi_1)N(0,1) + \pi_1N(\mu_1,1)$. The locations of the ideal change-points are shown in Figure \ref{fig:true_cploc} for various values of $\pi_1$ and $\mu_1$. As expected, stronger signal gives better separation between the two distributions resulting in the change-point approaching the unknown proportion. Moreover, the numerical results show that Assumption (A2) holds for almost all of the cases where $\mu_1\ge 2$ and $\pi_1 \le 0.3$, $\tilde{t}_\alpha < 0.5$ for both $\alpha = 1/2$ and $\alpha = 1$. However, for $\alpha = 1/2$ and $\mu_1 = 2$ with $\pi_1 \ge 0.29$, it appears that Assumption (A2) does not hold, suggesting that in this case, asymptotically, our method reduces to Storey's estimator with $\lambda = p_{(n/2)}$, as proposed in \cite{benjamini2006adaptive}. 
We note that these conditions on $\mu_1$ and $\pi_1$ do not depend on the sample size $n$.

\begin{figure}[H]
\begin{tabular}{cc}
\centering
\subfloat{\includegraphics[width=.4\textwidth]{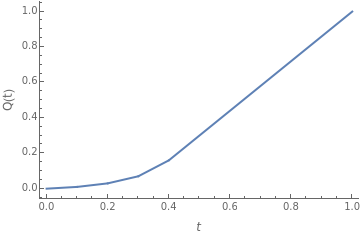}} &
\subfloat{\includegraphics[width=.4\textwidth]{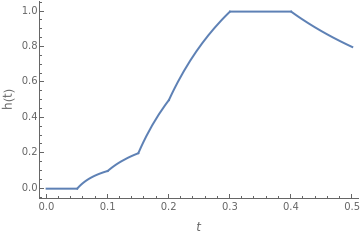}} 
\end{tabular}
\caption{An example where Assumption (A2) is violated as $h$ is constant on an interval. The $p$-value distribution is modeled as a mixture of several uniforms. Left: Quantile function $F^{-1}$  of the uniform mixture. Right: The function $h$ is constant on $[0.3, 0.4]$ where it reaches its maximum value.}
\label{fig:counterexample_unif}
\end{figure}
\begin{figure}[H]
\begin{tabular}{cc}
\centering
\subfloat{\includegraphics[width=.4\textwidth]{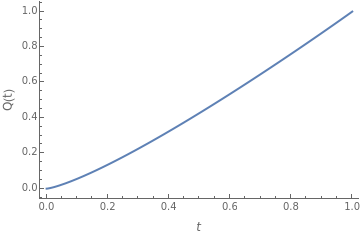}} &
\subfloat{\includegraphics[width=.4\textwidth]{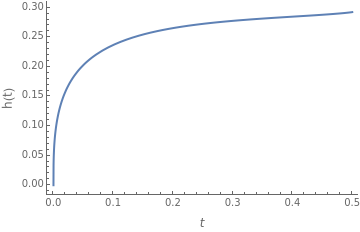}} 
\end{tabular}
\caption{An example where Assumption (A2) is violated as $h$ increases over $[0, 0.5]$. Left: Quantile function $F^{-1}$ of one-sided $p$-values from the Gaussian model $\pi_1N(\mu_1, 1) + \pi_0N(0,1)$ with $\pi_1 = 0.2$ and $\mu = 1$. Right: Corresponding $h$ function.}
\label{fig:counterexample_gaus}
\end{figure}

\begin{figure}[h]
\begin{tabular}{cc}
\subfloat{\includegraphics[width=.45\textwidth]{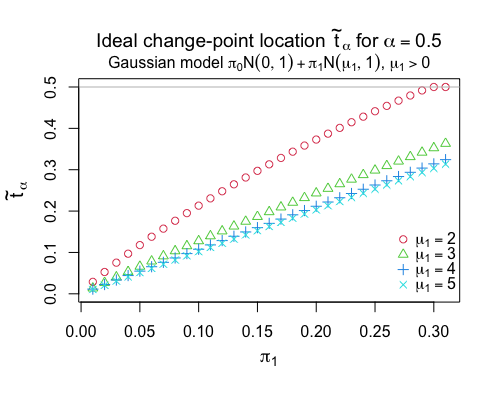}} &
\subfloat{\includegraphics[width=.45\textwidth]{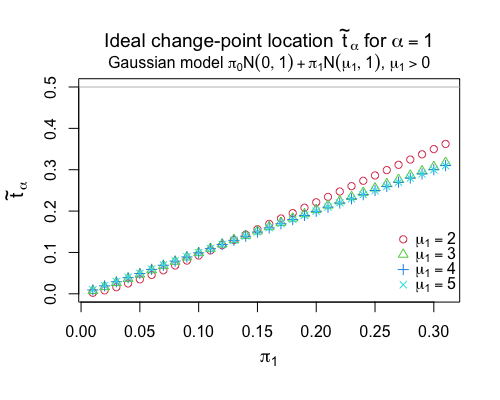}}
\end{tabular}
\caption{The points on the plots represent the locations of the asymptotic change-point $\tilde{t}_\alpha$ in the quantile function for the Gaussian model $\pi_0N(0,1) + \pi_1N(\mu_1,1)$, $\mu_1 > 0$, for $\alpha = 1/2$ (left) and $\alpha = 1$ (right) and various values of $\pi_1$ ($x$-axis) and different symbols representing the values of $\mu_1$.}
\label{fig:true_cploc}
\end{figure}

\end{remark}

\section{Simulations}
\label{sec:simulations}

In this section, we assess the performance of the DOS-Storey methods ($\hat{\pi}_1^1, \hat{\pi}_1^{1/2}$ and $\hat{\pi}_1^{a}$) by comparing them with various proportion estimators from the literature in terms of their bias, standard deviation (SD), and root mean squared error (RMSE).

The simulation setting considers the  Gaussian mean testing problem, $H_0: \mu = 0$ against $H_1: \mu >0$. The test statistics $T_i$ follow a $N(0,1)$ distribution under the null hypotheses and $N(\mu_1,1)$ with $\mu_1>0$ under the alternative. One-sided $p$-values are calculated from the test statistics as $p_i = 1 - \Phi(T_i)$ where $\Phi$ is the standard Gaussian CDF. We consider a fixed proportion of false null hypotheses. That is, for a sample of size $n$, and a given false null proportion $\pi_1$, the number of false null test statistics is set to $\lfloor n\pi_1 \rfloor$.

The adaptive estimator $\hat{\pi}_1^a$ depends on a threshold $\tau$, as described in Algorithm \ref{alg:dos_adaptive}. After experimenting with various constant values for $\tau$, we found that setting $\tau = n^{-1/2}$ best adapts to sparse and dense cases. This threshold, commonly regarded as the boundary for sparsity in theoretical analyses (e.g., \cite{Meinshausen2006}), will be used in the simulations below.
 
The simulations in Section C.1 of the Supplementary Material show that excluding the values from the beginning of the sequence $d(i)$ does not affect the estimates in practice. Therefore, we do not exclude any values when computing the estimates using the DOS method, i.e. we set $nc_n = 1$.

Below, we list and briefly describe the methods used in the simulation study.
\begin{enumerate}
    \item STS -- Storey and Tibshirani's \textit{smoother} method from \cite{storey2003statistical}, implemented in the \code{R} package \code{qvalue} by \cite{pack_qvalue}
    \item MGF -- Moment Generating Function method by \cite{broberg2005comparative}, implemented in the \code{R} package \code{SAGx} by \cite{pack_SAGx}
    \item LLF -- Langaas-Lindqvist-Ferkingstad by \cite{langaas2005estimating} 
    \item LSL -- Lowest Slope estimator by \cite{benjamini2000adaptive}
    \item MR -- Meinshausen-Rice by \cite{Meinshausen2006}
    \item JD -- Jiang-Doerge by \cite{Jiang2008}
    \item ST-MED -- Storey's estimator (\ref{eq:ss_pi1_est}) with  $\lambda=p_{(n/2)}$, as proposed in \cite{benjamini2006adaptive}
    \item ST-1/2 -- Storey's estimator (\ref{eq:ss_pi1_est}) with  $\lambda=1/2$
\end{enumerate}
Among the methods listed above, Storey-based methods include LSL, JD, ST-MED, and ST-1/2. LSL aims to identify the onset of the linear part and results in a conservative estimator. JD uses bootstrap and averages Storey's proportion estimator across several $\lambda$ values.
The statistical literature on adaptive FDR control typically recommends using ST-1/2  \citep{Blanchard2009, pmlr-v48-lei16, Ignatiadis2021}. Additionally, we have LLF, a density estimation-based method, and  MR, a consistent estimator constructed using the empirical processes theory. MGF is a moment generating function-based method that accounts for the behavior under the alternative. STS uses spline smoothing to combine the information from several $\lambda$ values close to 1. The implementation of STS within the function \code{pi0est} from the \code{R} package \code{qvalue} is not suitable for small sample sizes and typically requires a sample size of $n \ge 200$. 

Table \ref{tab:dos_sim_n1000} provides a comparison of various proportion estimators for a sample size $n=1000$ and different values of $\mu_1$ and $\pi_1$, based on $N=1000$ repetitions.
The results show that, in terms of the RMSE, $\hat{\pi}_1^1$ performs better in sparse cases, whereas $\hat{\pi}_1^{1/2}$ is better suited when there is a higher proportion of false nulls. The aDOS-Storey estimator $\hat{\pi}_1^a$ with $\tau = n^{-1/2}$ successfully adapts to the sparsity level, resulting in the overall best performing estimator.

Simulation results for small sample sizes $n=50$ and $n=100$ are presented in Table \ref{tab:small_sample_dos}. The STS method is excluded as it requires a larger sample size to compute the estimates. The results indicate that both $\hat{\pi}_1^1$ and $\hat{\pi}_1^{1/2}$ exhibit one of the smallest RMSE values among the considered estimators. It remains that $\hat{\pi}_1^1$ performs better in sparser cases, while $\hat{\pi}_1^{1/2}$  performs better in dense cases. Among the competing methods $\hat{\pi}_1^{a}$ behaves the best overall for all sparsity levels. 

\begin{table}[]
{\footnotesize
\begin{tabular}{lrrrrrrrrrrr}
     & $\hat{\pi}_1^1  $              & $\hat{\pi}_1^{1/2}$            & $\hat{\pi}_1^a $             & ST-1/2 & ST-MED              & JD    & LLF  & LSL                & MGF                 & MR    & STS   \\ \hline
\multicolumn{12}{c}{$\mu_1 = 3.5,  \pi_1 = 0.01, n_1 = 10$}                                                                                                                     \\ \hline
BIAS & -0.6                & 9.6                 & -0.6                & 8.4    & 8.0                 & 7.8   & 17.7 & -3.8               & 3.8                 & -4.6  & 32.0  \\
SD   & 3.8                 & 15.7                & 3.8                 & 22.0   & 20.5                & 21.2  & 22.7 & 2.5                & 14.4                & 5.7   & 55.2  \\
RMSE & {\ul \textbf{3.8}}  & 18.4                & {\ul \textbf{3.8}}  & 23.5   & 22.0                & 22.6  & 28.8 & {\ul \textbf{4.6}} & 14.8                & 7.3   & 63.8  \\ \hline
\multicolumn{12}{c}{$\mu_1 = 3.5, \pi_1 = 0.03, n_1 = 30$}                                                                                                                      \\ \hline
BIAS & -2.9                & 6.6                 & -2.2                & 3.2    & 3.3                 & 2.3   & 16.9 & -7.4               & 0.9                 & -8.5  & 23.8  \\
SD   & 5.6                 & 14.1                & 7.5                 & 26.7   & 25.2                & 26.1  & 22.1 & 3.5                & 17.5                & 6.4   & 61.3  \\
RMSE & {\ul \textbf{6.3}}  & 15.6                & {\ul \textbf{7.8}}  & 26.9   & 25.4                & 26.2  & 27.8 & 8.2                & 17.6                & 10.7  & 65.8  \\ \hline
\multicolumn{12}{c}{$\mu_1 = 3.0, \pi_1 = 0.05, n_1 = 50$}                                                                                                                      \\ \hline
BIAS & -8.6                & 4.8                 & 2.0                 & 1.1    & 1.4                 & -0.5  & 16.6 & -17.4              & -0.5                & -16.1 & 17.9  \\
SD   & 8.4                 & 16.2                & 16.8                & 29.3   & 27.3                & 29.3  & 24.3 & 5.4                & 18.2                & 8.1   & 67.0  \\
RMSE & {\ul \textbf{12.0}} & {\ul \textbf{16.9}} & {\ul \textbf{16.9}} & 29.3   & 27.4                & 29.3  & 29.4 & 18.2               & 18.2                & 18.1  & 69.3  \\ \hline
\multicolumn{12}{c}{$\mu_1 = 2.0, \pi_1 = 0.1, n_1 = 100$}                                                                                                                      \\ \hline
BIAS & -36.8               & -3.7                & -5.7                & -3.8   & -4.8                & -7.1  & 12.8 & -67.7              & -14.0               & -44.4 & 8.0   \\
SD   & 19.4                & 24.1                & 26.3                & 29.9   & 26.6                & 31.8  & 31.1 & 9.2                & 18.3                & 14.6  & 77.7  \\
RMSE & 41.6                & {\ul \textbf{24.3}} & 26.9                & 30.2   & 27.0                & 32.6  & 33.6 & 68.3               & {\ul \textbf{23.0}} & 46.7  & 78.1  \\ \hline
\multicolumn{12}{c}{$\mu_1 = 3.0, \pi_1 = 0.1, n_1 = 100$}                                                                                                                      \\ \hline
BIAS & -13.8               & 0.9                 & 0.9                 & 0.5    & 0.6                 & -2.8  & 16.6 & -28.1              & -1.3                & -22.9 & 7.4   \\
SD   & 10.5                & 16.8                & 16.8                & 29.7   & 26.6                & 31.6  & 24.9 & 7.8                & 18.0                & 8.9   & 77.6  \\
RMSE & {\ul \textbf{17.4}} & {\ul \textbf{16.9}} & {\ul \textbf{16.9}} & 29.7   & 26.6                & 31.8  & 29.9 & 29.2               & 18.0                & 24.6  & 77.9  \\ \hline
\multicolumn{12}{c}{$\mu_1 = 2.0, \pi_1 = 0.2, n_1 = 200$}                                                                                                                      \\ \hline
BIAS & -48.9               & -16.1               & -16.1               & -8.7   & -13.3               & -11.9 & 9.6  & -117.8             & -28.7               & -62.4 & 3.0   \\
SD   & 25.7                & 23.2                & 23.2                & 28.4   & 22.4                & 31.5  & 33.8 & 15.9               & 17.3                & 17.4  & 81.2  \\
RMSE & 55.2                & {\ul \textbf{28.2}} & {\ul \textbf{28.2}} & 29.7   & {\ul \textbf{26.0}} & 33.6  & 35.1 & 118.9              & 33.5                & 64.7  & 81.3  \\ \hline
\multicolumn{12}{c}{$\mu_1 = 3.0, \pi_1 = 0.2, n_1 = 200$}                                                                                                                      \\ \hline
BIAS & -20.6               & -2.9                & -2.9                & -0.2   & 0.1                 & -4.0  & 15.9 & -44.0              & -3.8                & -31.5 & 1.5   \\
SD   & 12.9                & 17.1                & 17.1                & 27.9   & 22.1                & 32.3  & 25.0 & 10.7               & 16.8                & 9.8   & 81.3  \\
RMSE & 24.3                & {\ul \textbf{17.3}} & {\ul \textbf{17.3}} & 27.9   & 22.1                & 32.5  & 29.6 & 45.3               & {\ul \textbf{17.2}} & 33.0  & 81.3  \\ \hline
\multicolumn{12}{c}{$\mu_1 = 3.0, \pi_1 = 0.3, n_1 = 300$}                                                                                                                      \\ \hline
BIAS & -24.0               & -6.4                & -6.4                & -0.3   & -2.7                & -3.5  & 15.3 & -54.8              & -6.8                & -37.7 & 2.6   \\
SD   & 13.6                & 15.1                & 15.1                & 26.2   & 16.7                & 31.3  & 25.1 & 13.7               & 15.9                & 10.1  & 77.3  \\
RMSE & 27.5                & {\ul \textbf{16.3}} & {\ul \textbf{16.3}} & 26.2   & {\ul \textbf{17.0}} & 31.5  & 29.3 & 56.5               & 17.3                & 39.0  & 77.30
\end{tabular}
}
\caption{Bias, standard deviation, and the RMSE of the estimated number of the false null hypotheses ($n\times \hat{\pi}_1$). The model is Gaussian, as described at the beginning of Section \ref{sec:simulations}. Simulations are conducted for various values of $\pi_1$ and $\mu_1$, with a sample size of $n=1000$ and 1000 repetitions. Bold and underlined values correspond to the two smallest RMSE in the row.}
\label{tab:dos_sim_n1000}
\end{table}


\begin{table}[]
{\footnotesize
\begin{tabular}{lrrrrrrrrrr}
     & $\hat{\pi}_1^1$                & $\hat{\pi}_1^{1/2}$               & $\hat{\pi}_1^a$               & ST-1/2 & ST-MED              & JD    & LLF   & LSL                 & MGF                 & MR     \\ \hline
\multicolumn{11}{c}{$\mu_1= 3 , \pi_1= 0.1 , n= 50      $}                                                                                                                 \\ \hline
BIAS & 0.90                & 2.50                & 1.10                & 0.80   & 0.70                & 0.90  & 3.80  & -2.00               & 0.00                & -3.20  \\
SD   & 2.70                & 3.10                & 2.90                & 5.60   & 4.50                & 5.20  & 5.50  & 1.90                & 3.60                & 1.50   \\
RMSE & {\ul \textbf{2.90}} & 4.00                & 3.10                & 5.70   & 4.60                & 5.20  & 6.70  & {\ul \textbf{2.80}} & 3.60                & 3.60   \\ \hline
\multicolumn{11}{c}{$\mu_1= 2 , \pi_1= 0.2 , n= 50      $}                                                                                                                 \\ \hline
BIAS & -0.90               & 0.60                & -0.40               & -0.30  & -0.60               & -0.60 & 3.10  & -4.80               & -1.50               & -6.50  \\
SD   & 3.60                & 3.30                & 3.90                & 6.00   & 4.50                & 5.60  & 6.20  & 2.70                & 3.70                & 2.20   \\
RMSE & {\ul \textbf{3.70}} & {\ul \textbf{3.30}} & 3.90                & 6.00   & 4.60                & 5.60  & 7.00  & 5.50                & 4.00                & 6.90   \\ \hline
\multicolumn{11}{c}{$\mu_1= 2 , \pi_1= 0.4 , n= 50      $}                                                                                                                 \\ \hline
BIAS & -3.20               & -2.50               & -2.50               & -0.80  & -2.60               & -1.60 & 2.60  & -6.60               & -2.90               & -9.50  \\
SD   & 2.60                & 2.30                & 2.30                & 5.60   & 2.80                & 5.50  & 6.40  & 3.80                & 3.40                & 2.60   \\
RMSE & 4.20                & {\ul \textbf{3.40}} & {\ul \textbf{3.40}} & 5.60   & {\ul \textbf{3.80}} & 5.80  & 6.90  & 7.70                & 4.50                & 9.90   \\ \hline
\multicolumn{11}{c}{$\mu_1= 3 , \pi_1= 0.05 , n= 100      $}                                                                                                               \\ \hline
BIAS & 0.80                & 3.80                & 0.90                & 1.80   & 1.60                & 1.70  & 5.40  & -2.20               & 0.40                & -3.30  \\
SD   & 3.20                & 4.80                & 3.40                & 7.40   & 6.40                & 6.80  & 7.30  & 1.90                & 4.70                & 1.80   \\
RMSE & {\ul \textbf{3.30}} & 6.10                & 3.60                & 7.60   & 6.60                & 7.00  & 9.10  & {\ul \textbf{3.00}} & 4.70                & 3.80   \\ \hline
\multicolumn{11}{c}{$\mu_1= 3 , \pi_1= 0.1 , n= 100      $}                                                                                                                \\ \hline
BIAS & 0.60                & 3.40                & 1.30                & 0.50   & 0.60                & 0.50  & 5.10  & -2.80               & -0.20               & -4.30  \\
SD   & 4.00                & 4.70                & 4.70                & 8.30   & 7.10                & 7.70  & 7.30  & 2.40                & 5.20                & 2.10   \\
RMSE & {\ul \textbf{4.00}} & 5.80                & 4.80                & 8.30   & 7.10                & 7.70  & 8.90  & {\ul \textbf{3.70}} & 5.20                & 4.80   \\ \hline
\multicolumn{11}{c}{$\mu_1= 2 , \pi_1= 0.2 , n= 100      $}                                                                                                                \\ \hline
BIAS & -3.00               & 0.30                & -0.60               & -0.90  & -1.50               & -1.50 & 4.20  & -9.40               & -2.90               & -11.00 \\
SD   & 5.90                & 5.10                & 6.20                & 8.90   & 6.70                & 8.60  & 9.10  & 3.90                & 5.20                & 3.50   \\
RMSE & 6.60                & {\ul \textbf{5.10}} & 6.20                & 9.00   & 6.90                & 8.70  & 10.00 & 10.20               & {\ul \textbf{6.00}} & 11.50  \\ \hline
\multicolumn{11}{c}{$\mu_1= 2 , \pi_1= 0.4 , n= 100      $}                                                                                                                \\ \hline
BIAS & -6.50               & -5.20               & -5.20               & -1.90  & -5.30               & -2.80 & 3.10  & -14.00              & -5.90               & -16.20 \\
SD   & 4.30                & 3.50                & 3.50                & 8.00   & 3.90                & 8.10  & 9.30  & 5.50                & 4.80                & 4.10   \\
RMSE & 7.80                & {\ul \textbf{6.20}} & {\ul \textbf{6.20}} & 8.20   & {\ul \textbf{6.60}} & 8.60  & 9.70  & 15.10               & 7.60                & 16.70 
\end{tabular}
}
\caption{Bias, standard deviation, and the RMSE of the estimated number of the false null hypotheses ($n\times \hat{\pi}_1$). The model is Gaussian, as described at the beginning of Section \ref{sec:simulations}. Simulations are conducted for various values of $n, \pi_1$ and $\mu_1$, and 1000 repetitions. Bold and underlined values correspond to the two smallest RMSE in the row.}
\label{tab:small_sample_dos}
\end{table}

\section{Real Data Example}
\label{sec:realdata}

Copy Number Variations (CNVs) are genetic alterations characterized by changes in the number of copies of specific DNA segments within an individual's genome, including duplications or deletions of these segments.  These variations play a crucial role in cancer development and progression, making their detection important for understanding the genetic causes of the disease. CNV data is typically obtained through aCGH (array comparative genomic hybridization), resulting in \textit{log ratio data}. In this data, no variation corresponds to a value of 0, while deletions are represented as decreases in value, and duplications as increases in value.

\begin{figure}[H]
    \centering
    \includegraphics[width=\textwidth]{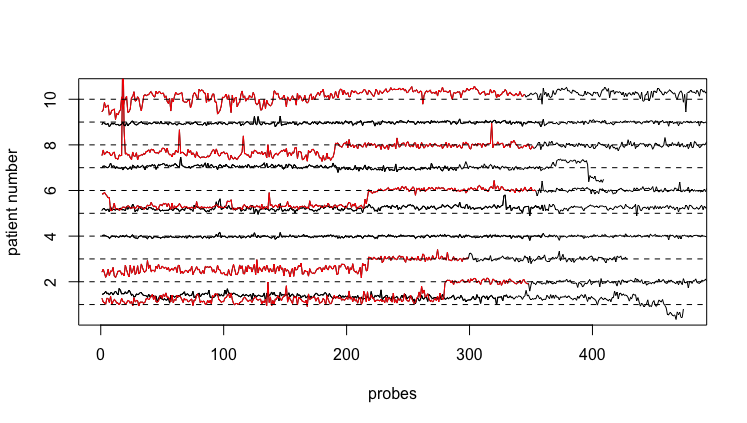}
    \caption{Log ratio data for chromosome 1, for the first ten patients. }
    \label{fig:cnv_patients}
\end{figure}

We analyze the dataset from the R package \code{neuroblastoma}, which includes annotated log ratio data for 575 patients with neuroblastoma. This data is initially considered in \cite{Hocking2013}. The dataset covers six chromosomes: 1, 2, 3, 4, 11, and 17. Within each chromosome, an interval of interest is examined. For each chromosome, an expert annotates each patient's data as either ``breakpoint'' or ``normal'' to indicate the suspected presence of changes. 
CNV data for the first chromosome for the first ten patients is shown in Figure \ref{fig:cnv_patients}. The probe locations are not aligned across patients, and the $x$-axis shows the index of the probe for which we have available data. Deletions can be seen in patients 2, 3, 6, and 8. We refer to these jumps in the piecewise constant mean of the sequences as breakpoints, not to be confused with the notion of change-points in the quantile function of the $p$-values used in the paper so far.

We explore the application of our method in the context of breakpoint inference. Our objective is to estimate the prevalence of copy number alterations (breakpoints) for each chromosome, for which we propose to use the DOS-Storey approach. Copy number log ratio data is usually assumed to be independent and Gaussian, with a piecewise constant underlying mean \citep{Zhang2010, jeng2012simultaneous}. However, the available data contains some outliers (see patient 8 in Figure \ref{fig:cnv_patients}), and for that reason before the analysis we trimmed the data by excluding the data points in the lower or upper 2.5th quantile. 
For each patient and each chromosome, a $p$-value arising from testing whether the sequence contains a breakpoint is obtained using the method by \cite{Jewell2022}. As a result we get a $p$-value for each patient in each of the six genomic regions. All $p$-values are shown in Figure \ref{fig:realdata_pvals}.
Note that the black bars correspond to the patients that are annotated as ``breakpoint'' by the expert, however the ground truth of whether the breakpoint is present is unknown. The critical component of this analysis is the computation of $p$-values using the method proposed in \cite{Jewell2022}. The estimation and inference on the breakpoints is performed using package \code{ChangepointsInference} \citep{pack_cpinference}, which enables estimation of the $p$-values using a post-selection inference approach.
A single breakpoint is estimated in each sequence using the CUSUM statistic. The fixed window parameter for testing the estimated breakpoint is set to $h = 5$. 

Using the three DOS-Storey estimators $\hat{\pi}_1^1, \hat{\pi}_1^{1/2}$ and $ \hat{\pi}_1^a$, for each chromosome we estimate the number of patients with breakpoints and compare it to the number of patients annotated as having a breakpoint. We compare the estimated values to those obtained by Storey's method with $\lambda = 0.5$ (ST-1/2). The results showing the estimated number of affected patients for each chromosome are presented in Table \ref{tab:realdata}. The values in the rightmost column (ANNOT) are the reported numbers of patients annotated as having a breakpoint. Although the ground truth is unknown, we observe that all three DOS method produce estimates that are closer to the annotated values. Particularly, $\hat{\pi}_1^1$ gives the best result for Chromosome 3, where the annotations indicate a small number of affected patients.

\begin{figure}[h]
\begin{tabular}{ccc}
\subfloat{\includegraphics[width=.32\textwidth]{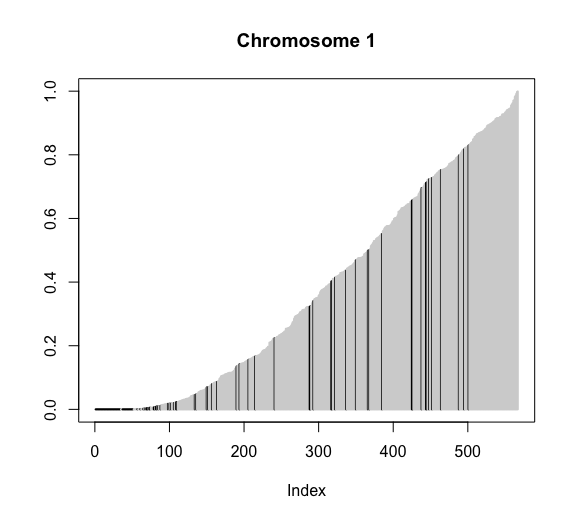}} &
\subfloat{\includegraphics[width=.32\textwidth]{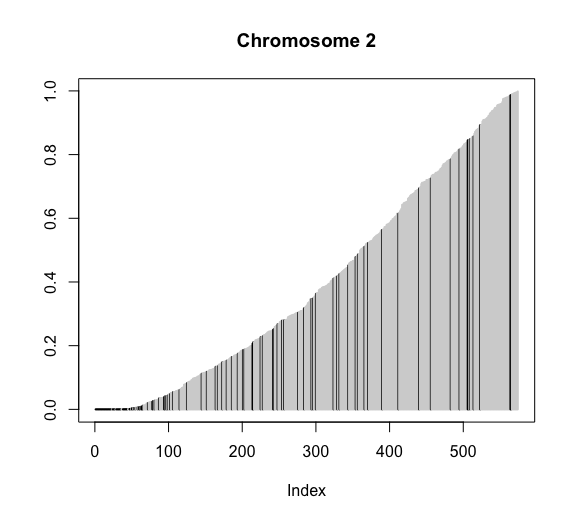}} &
\subfloat{\includegraphics[width=.32\textwidth]{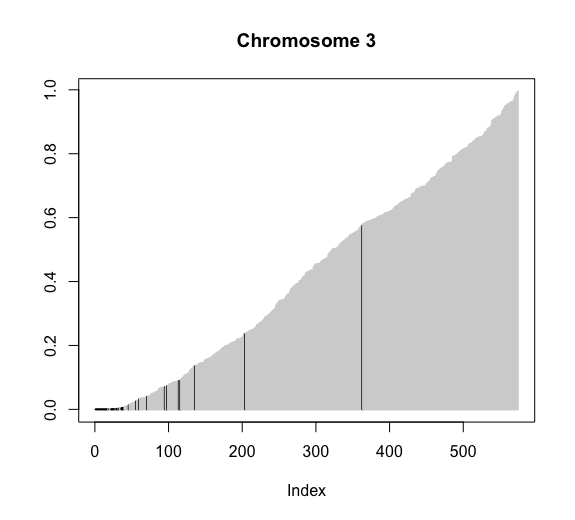}} \\
\subfloat{\includegraphics[width=.32\textwidth]{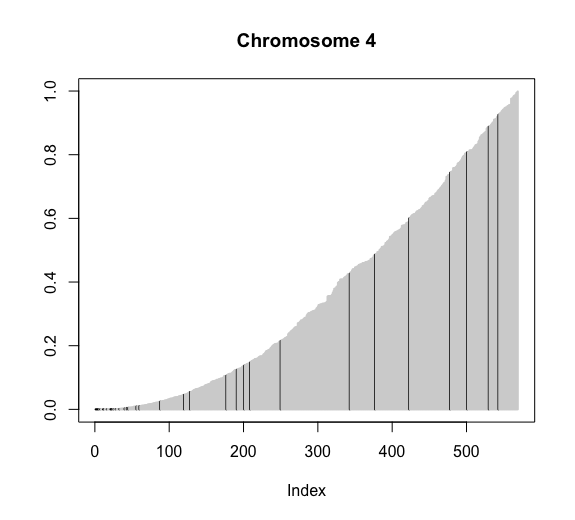}}  &
\subfloat{\includegraphics[width=.32\textwidth]{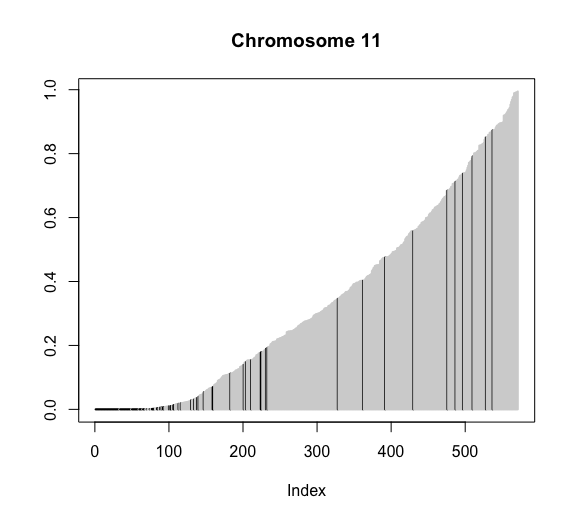}} &
\subfloat{\includegraphics[width=.32\textwidth]{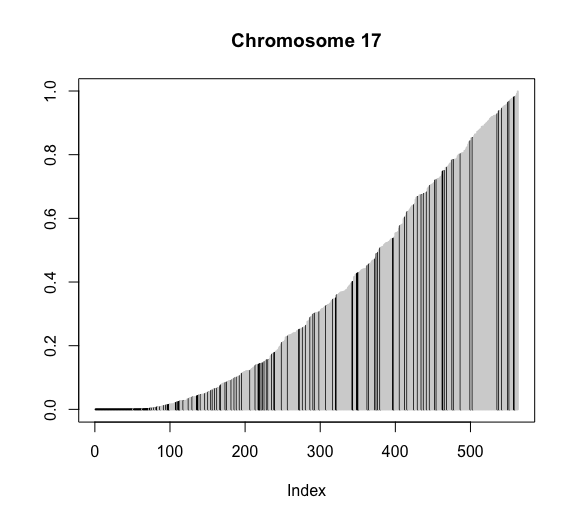}} 
\end{tabular}
\caption{$p$-value plots for the neuroblastoma data. The black bars denote the $p$-values corresponding to the sequences that are annotated as ``breakpoint'' (with change-point in the copy number data), while the gray background corresponds to the sequences annotated as ``normal'' (no change in the copy number values).}
\label{fig:realdata_pvals}
\end{figure}       

\begin{singlespacing}
\begin{table}[H]
\centering
\begin{tabular}{c|ccccc}
       & $\hat{\pi}_1^1$ & $\hat{\pi}_1^{1/2}$ & $\hat{\pi}_1^{a}$ & ST-1/2 & ANNOT \\ \hline
Chr 1  & 150 & 149 & 149     & 167   & 103   \\ \hline
Chr 2  & 147 & 147 & 147   & 146   & 110   \\ \hline
Chr 3  & 39   & 84 & 84   & 70    & 43    \\ \hline
Chr 4  & 168 & 168 & 168 & 195   & 35    \\ \hline
Chr 11 & 174  & 174 & 174  & 243   & 107   \\ \hline
Chr 17 & 171  & 171 & 171  & 193   & 175  
\end{tabular}
\caption{The number of affected sequences for each chromosome, estimated using the three DOS-Storey estimators  and the ST-1/2 estimator, along with the number of sequences annotated as ``breakpoint'' (ANNOT).}
\label{tab:realdata}
\end{table}
\end{singlespacing}


\begin{singlespacing}
    
\bibliographystyle{jabes}
\bibliography{main.bib}
\end{singlespacing}

\end{document}


\maketitle

This supplement provides further insights into the theoretical properties of the proposed method and presents additional simulation results. We outline the sections as follows:
\begin{itemize}
    \item In Section \ref{sec:proofs} we provide the proof of Theorem 1 from the main paper.
    \item Section \ref{supp:simul_gaussian_mixture} contains numerical results on the limiting behavior of the proposed statistics under the Gaussian model computed using Mathematica.
    \item Section \ref{supp:simul} contains various additional simulation results. Namely, it involves: 
        \begin{enumerate}
            \item Simulation results examining the sensitivity of the DOS-Storey estimator with respect to $c_n$, the proportion of excluded values, in Section \ref{supp:simul_cn}. 
        
           \item Comparing the DOS-Storey methods to the other change-point based methods for proportion estimation in Section \ref{supp:simul_cp_methods}.
            \item Comparison of various proportion estimators for adaptive FDR control using the adaptive BH procedure in Section \ref{supp:simul_fdr}.

         \item Simulation results under dependence in Section \ref{supp:simul_dependence}
            
        \end{enumerate}
       \item Section \ref{supp:interpr}  contains additional discussion on the motivation behind the DOS statistic.   
\end{itemize}

\appendix

\section[]{Proofs}
\label{sec:proofs}
Before presenting the main theorems we state two lemmas that connect the quantile process of the $p$-value distribution to the uniform quantile process. This will allow us to later use some existing results on the almost sure behaviour of the weighted uniform quantile process in the proof of Theorem 1.

\subsubsection*{Notation}

We denote with $\hat{E}_n$ the empirical CDF of a sample of size $n$ from $U[0,1]$ distribution. The corresponding empirical and quantile processes are respectively:
\begin{align*}
   \alpha_n(y)&=\sqrt{n}(\hat{E}_n(y)-y), \\
    u_n(y)&=\sqrt{n}(\hat{E}_n^{-1}(y)-y).
\end{align*}
For a sample of random variables $X_1, \dots, X_n$ with CDF $F$, we denote its empirical CDF as $\hat{F}_n$, and the corresponding quantile process as 
\begin{equation*}
     q_n(y)=\sqrt{n}(\hat{F}_n^{-1}(y)-F^{-1}(y)).
\end{equation*}

\subsection{Some Useful Lemmas}

\begin{lemma} \label{lemma:uniform_mix_quantile_process}
Let $X_1,\dots, X_n$ be the sample from a distribution with CDF given as
\begin{equation*}
    F(x)=\pi_1F_1(x)+\pi_0x,
\end{equation*}
where $F_1$ is a continuous weakly concave function. 
It holds that
\begin{equation} \label{eq:quant_unifquant_rel}
    q_n(y)\le \frac{1}{\pi_0}u_n(y), \quad y\in (0,1).
\end{equation}

\end{lemma}
\begin{proof}

Let $X_{k;n}$ be the $k$th order statistic of the sample and let
$(k-1)/n<y\le k/n$. It holds that
\begin{align*}
    q_n(y)&=\sqrt{n}(\hat{F}_n^{-1}(y)-F^{-1}(y))\\
        &= \sqrt{n}(X_{k;n}-F^{-1}(y))\\
        &=\sqrt{n}(F^{-1}(F(X_{k;n}))-F^{-1}(y))\\
        &=\sqrt{n}(F^{-1}(U_{k;n})-F^{-1}(y))\\
        &=\sqrt{n}\frac{F^{-1}(U_{k;n})-F^{-1}(y)}{U_{k;n}-y}(U_{k;n}-y)\\
        &= \frac{F^{-1}(U_{k;n})-F^{-1}(y)}{U_{k;n}-y} \sqrt{n}(\hat{E}_n^{-1}(y)-y)\\
        &\le \sup_{x,y} \frac{F^{-1}(y)-F^{-1}(x)}{y-x} u_n(y) \\
        &\le \frac{1}{\pi_0}u_n(y).
\end{align*}

The last inequality follows from the concavity assumption as follows. As $F_1$ is a concave function on $[0,1]$, and $F$ is a linear combination of $F_1$ and a linear function, $F$ is also a concave function on $[0,1]$. It holds that the inverse of a continuous, concave and increasing function on an interval is convex on the same interval, so it follows that $F^{-1}$ is a convex function.
Since $F(x)\le \pi_1+\pi_0x$ for $x\in [0,1]$ 
it holds that for any $0\le x<y\le1$
\begin{align*}
  \frac{F^{-1}(y)-F^{-1}(x)}{y-x} &\le \frac{1-F^{-1}(x)}{1-x}\\
    &\le \frac{1}{\pi_0}.
\end{align*}
\end{proof}

\begin{lemma} \label{lemma:dkw_unif_mix_quantile_process}
Under the assumptions of Lemma \ref{lemma:uniform_mix_quantile_process}, it holds that
\begin{equation*}
    P\left(\sup_{0<y<1} |q_n(y)|\ge x\right)\le 2\exp(-2x^2/C^2).
\end{equation*}

\end{lemma}
\begin{proof}
To prove this we use the result from Lemma \ref{lemma:uniform_mix_quantile_process} and the relationship between the uniform empirical and the uniform quantile process.
By the change of variable argument we have  $\sup_{0<y<1} |u_n(y)|=\sup_{0<y<1} |\alpha_n(y)|$ (see Remark 1.4.1 in \cite{Csorgo1983}). The result now follows from the Dvoretzky-Kiefer-Wolfowitz inequality, using the tight bound from \cite{Massart1990}:
\begin{align*}
P\left(\sup_{0<y<1} |q_n(y)|\ge x\right) &\le P\left(\sup_{0<y<1} |u_n(y)|C\ge x\right)\\
&= P\left(\sup_{0<y<1} |\alpha_n(y)|\ge x/C\right)\\
&\le 2 \exp(-2x^2/C^2).
\end{align*}
\end{proof}

\subsection{Proof of Theorem 1}

\begin{proof}
Let 
\begin{equation}
    h_{n}(t)\vcentcolon=\frac{F_n^{-1}(2t)-2F_n^{-1}(t)}{t^\alpha}
\end{equation}
The empirical function $h_{n}(t)$ is approximating the ideal function $h(t)$ defined in Eq. (6) of the main paper as
\begin{equation*}
   h(t)\vcentcolon=\frac{F^{-1}(2t)-2F^{-1}(t)}{t^\alpha}
\end{equation*}
For the sake of notation simplicity, we suppress the dependence of $h(t)$ and $h_n(t)$ on $\alpha$ and $F$.
For the DOS sequence it holds that $d_\alpha(i)=h_n(i/n)$, for $i=1,\dots,n$.
Function $h$ is positive on $(0,1)$ because of the convexity assumption on $F^{-1}$.
We start with the following sequence of inequalities, aiming to upper bound the rate of difference $ \abs{ h_{n}(t) - h(t)}$, uniformly for $t\in (c_n,1)$, using strong limit theorems for weighted uniform quantile processes.
\begin{align}
    \abs{ h_{n}(t) - h(t)}&=
    \left| \frac{F_n^{-1}(2t)-2F_n^{-1}(t)}{t^\alpha} -  \frac{F^{-1}(2t)-2F^{-1}(t)}{t^\alpha}\right|\\
    &\le 2^\alpha \left| \frac{F_n^{-1}(2t)-F^{-1}(2t)}{(2t)^\alpha}\right|+2 \left| \frac{F_n^{-1}(t)-F^{-1}(t)}{t^\alpha}\right|\\
    &\le \frac{2^\alpha}{\sqrt{n}}\frac{|q_n(2t)|}{(2t)^\alpha}+\frac{2}{\sqrt{n}}\frac{|q_n(t)|}{t^\alpha}\\
    &\le\frac{2^\alpha + 2}{\sqrt{n}}\sup_{t\in (c_n,1)} \frac{|q_n(t)|}{t^\alpha}\\
    &\le  \frac{2^\alpha + 2}{\sqrt{n}}\frac{1}{\pi_0}\sup_{t\in (c_n,1)} \frac{|u_n(t)|}{t^\alpha} \label{eq:unif_bound_func}
\end{align}   
In the last inequality we used Lemma \ref{lemma:uniform_mix_quantile_process}.
To bound the weighted uniform quantile process we use Theorem 2 case (III) from \cite{Einmahl1988}, setting $\nu = 0, a_n = \log(n)/n$ using the notation therein, that describes its almost sure behaviour
\begin{equation} \label{eq:einmahl_mason}
    \limsup_{n\to \infty} \sup_{c_n\le t \le 1} \frac{c_n^{\alpha-1/2}|u_n(t)|}{t^{\alpha}\sqrt{\log\log n}}\stackrel{a.s.}{=}\sqrt{2}.
\end{equation}
meaning that, for any $\eps>0$ and large enough $n$, on a set of probability 1, it holds that
\begin{equation} \label{eq:scaled_uqp_ineq}
    \frac{1}{\sqrt{n}}\sup_{c_n\le t \le 1}\frac{|u_n(t)|}{t^\alpha} \le \frac{\sqrt{\log\log n}}{\sqrt{n}c_n^{\alpha - 1/2}}(\sqrt{2}+\eps)
\end{equation}
Finally, (\ref{eq:scaled_uqp_ineq}) and  (\ref{eq:unif_bound_func}) give a uniform upper bound for $|h_n(t)-h(t)|$ on $t\in (c_n,1]$
\begin{equation} \label{eq:uniform_bound_for_difference_h}
    \sup_{t\in (c_n,1]}|h_n(t)-h(t)|\le C\frac{\sqrt{\log\log n}}{\sqrt{n}c_n^{\alpha - 1/2}},
\end{equation}
where $C$ is a constant that for large $n$ approaches $\sqrt{2}$. Denote
\begin{align*}
    \hat{t}^{d}_n&\defeq \frac{1}{n}\argmax_{1\le i\le n} d(i) &\hat{t}_{n} \defeq \argsup_{t\in [0,1]} h_{n}(t) \\
    \tilde{t}^d_n &\defeq \frac{1}{n}\argmax_{1\le i\le n} h(i/n) &
    \tilde{t} \defeq \argmax_{t\in [0,1]} h(t)
\end{align*}
For simplicity, we suppress the dependence of these quantities on $\alpha$ and $F$.
Since $\tilde{t}^d_n$ is the argmax of a continuous function on an increasingly dense grid, it holds that $|\tilde{t}^d_n-\tilde{t}|\le 1/n$. Since $h$ is bounded and continuous on $[0,1]$, it follows that for some $C'>0$ $|h(\tilde{t}^d_n)-h(\tilde{t})|\le C'/n$.
For $n$ large enough, $\tilde{t}>c_n$, the following sequence of inequalities holds with probability 1.
\begin{align*}
    h(\tilde{t})
    &\le  h(\tilde{t}^d_n)+ C'/n\\
    &\le  |h_n(\tilde{t}^d_n)|+ C\frac{\sqrt{\log\log n}}{\sqrt{n}c_n^{\alpha - 1/2}}+C'/n\\
    &\le |h_n(\hat{t}^d)|+ C_1\frac{\sqrt{\log\log n}}{\sqrt{n}c_n^{\alpha - 1/2}}\\
    &\le h(\hat{t}^d)+2 C_1\frac{\sqrt{\log\log n}}{\sqrt{n}c_n^{\alpha - 1/2}},
\end{align*}
where $C_1$ gets arbitrarily close to $\sqrt{2}$.
It implies
\begin{equation} \label{eq:f_small_diff_nonunif}
    h(\tilde{t})-h(\hat{t}^{d})\le  2C_1\frac{\sqrt{\log\log n}}{\sqrt{n}c_n^{\alpha - 1/2}}
\end{equation}
We prove the consistency of $\hat{t}^d$  by contradiction. 
However, the rate of convergence depends on the differentiability of $h$, and we separate three different cases. 
\begin{description}
    \item[Case 1:] $h$ has a second derivative at $\tilde{t}$, and  $h''(\tilde{t})\neq 0$.
    \item[Case 2:] $h$ has a second derivative at $\tilde{t}$, and  $h''(\tilde{t})= 0$.
    \item[Case 3:] $h$ does not have a second derivative at $\tilde{t}$.
\end{description}
We start with Case 1, and note that a sufficient condition for $h$ to be twice differentiable is that $F$ is twice differentiable on $(0,1)$.
Assuming that $|\hat{t}^d-\tilde{t}|> \frac{\sqrt{\log\log n}}{\sqrt{n}c_n^{\alpha - 1/2}}$, it holds that
\begin{align*}
    \left|h(\tilde{t})-h(\hat{t}^d)\right|&= (\hat{t}^d-\tilde{t})^2|h''(\tilde{t})|+ o((\hat{t}^d-\tilde{t})^2)\\
    &\ge C_2\frac{\log\log n}{\sqrt{n}c_n^{\alpha - 1/2}}.
\end{align*}
For large $n$, the last inequality is in contradiction with (\ref{eq:f_small_diff_nonunif}), so it must hold that
\begin{equation} \label{eq:rate_differentiable}
    \left|\hat{t}^d-\tilde{t}\right|\le \frac{\sqrt{\log\log n}}{\sqrt{n}c_n^{\alpha - 1/2}},
\end{equation}
which proves the consistency stated in Eq. (7) in the main paper. 
For Case 2, if $h''(\tilde{t})=0$, the consistency still holds, since not all derivatives can be zero, but the rate of convergence is slower accordingly.
In Case 3, when $h$ is not differentiable at $\tilde{t}$, such that left and right derivatives at $\tilde{t}$ are not equal, but lower bounded by a constant larger than zero in an interval around $\tilde{t}$, we can get a better convergence rate:
\begin{equation*}
    \left|\hat{t}^{d}-\tilde{t}\right|\le C\frac{\log\log n}{nc_n^{2\alpha - 1}}.
\end{equation*}
This is the case for example when $F$ is a mixture of uniform distributions (see Corollary 1). In other cases, we similarly have (\ref{eq:rate_differentiable}) to hold.
We proceed under Case 1, assuming that $h''(\tilde{t})\neq 0$ holds, while the results for other cases can be obtained similarly. The following sequence of inequalities holds almost surely and proves the almost sure convergence in (8):
\begin{align}
    \abs{p_{(\hat{k}_{\alpha})}-F^{-1}(\tilde{t})}&= \abs{\hat{F}_n^{-1}(\hat{t}^d)-F^{-1}(\tilde{t})}\nonumber \\
    &\le \abs{\hat{F}_n^{-1}(\hat{t}^d)-F^{-1}(\hat{t}^d)}+ \abs{F^{-1}(\hat{t}^d)-F^{-1}(\tilde{t})} \label{eq:b_consistency_terms2}\\
    &\le C_3 \sqrt{\frac{\log\log n}{n}}+ C_4\sqrt{\frac{\log\log n}{nc_n^{2\alpha - 1}}} \nonumber \\
    &\le C_5\sqrt{\frac{\log\log n}{nc_n^{2\alpha - 1}}}. \nonumber
\end{align}
For the first term in (\ref{eq:b_consistency_terms2}) we use Lemma 1 and then the Chung-Smirnov law of iterated logarithm, to get the inequality which holds almost surely. 
For the second term we use the fact that $F^{-1}$ is Lipschitz continuous, and the obtained rate of convergence in (\ref{eq:rate_differentiable}).
The convergence in Eq. (9) of the main paper follows similarly:
\begin{align*}
    \abs{\hat{\pi}^{\alpha}_1-\frac{\tilde{t}-F^{-1}(\tilde{t})}{1-F^{-1}(\tilde{t})}} &= \abs{\frac{\hat{t}^d-F_n^{-1}(\hat{t}^d)}{1-F_n^{-1}(\hat{t}^d)}- \frac{\tilde{t}-F^{-1}(\tilde{t})}{1-F^{-1}(\tilde{t})}}\\
    &\le C_6\sqrt{\frac{\log\log n}{n^{\frac{1-\theta}{2}}}}.
\end{align*}
\end{proof}

\begin{remark} \normalfont
\label{rem:cn_exclude}
The conditions on $c_n$ in Theorem 1 come from the theorem in \cite{Einmahl1988}. Instead of $c_n \to 0$, we can trivially take $c_n= \eps\in (0,1)$, in which case we can simply use the Chung-Smirnov law stated in the proof of Theorem 1, to get the rate of convergence of $\frac{\sqrt{\log\log n}}{\sqrt{n}}$.
From the work of \cite{Einmahl1988} we see that the choice of $c_n$ is very important in the theory of uniform quantile processes. 
As the false-null distribution $F_1$ is unknown, using the result from Lemma \ref{lemma:dkw_unif_mix_quantile_process} we bound the quantile process of distribution $F$ by a uniform quantile process. This approximation is convenient as most of the results in the theory of quantile processes are given only for the uniform quantile process.
However, the behaviour of the weighted uniform quantile process around 0 will be more variable than that of weighted quantile process of a distribution $F$. $F$ is more concentrated around zero and the sample quantiles will be closer to the true quantiles than in the case of uniform distribution, which reduces the boundary problem. 
\end{remark}

\section{Numerical Results for the Gaussian Mixture Model} 
\label{supp:simul_gaussian_mixture}

In this section, we build upon Remark 1 concerning the ideal asymptotic quantities $\tilde{t}_\alpha$ and $\tilde{\pi}_1^\alpha$ defined in Theorem 1 under Gaussian mixture model. 
To gain insights into the behavior of these asymptotic quantities, we performed computations of $\tilde{t}_\alpha$ and $\tilde{\pi}_1^\alpha$ under the Gaussian model. The results concerning $\tilde{t}_\alpha$ were shown in Remark 1 in the main paper, while here we show the results for $\tilde{\pi}_1^\alpha$. $\tilde{\pi}_1^\alpha$ is generally smaller than the true proportion, and we call this quantity the estimable proportion. 

The numerical results are obtained using \code{Mathematica} and can be found within the \code{MTCP} package. 
Figure \ref{fig:estimable_proportion} shows $\tilde{\pi}_1^\alpha/\pi_1$ as a function of $\pi_1$ and for different values of the nonzero mean $\mu_1$. 
This quantity is shown for the two DOS-Storey estimators, for $\alpha = 1/2$ (left-hand plot) and $\alpha = 1$ (right-hand plot). $\tilde{\pi}_1^\alpha/\pi_1$ represents the ratio of the estimable proportion and the true proportion under this model. For larger signal values, the estimable proportion is close to the true proportion.

\begin{figure}[H]
\begin{tabular}{cc}
\subfloat{\includegraphics[width=.45\textwidth]{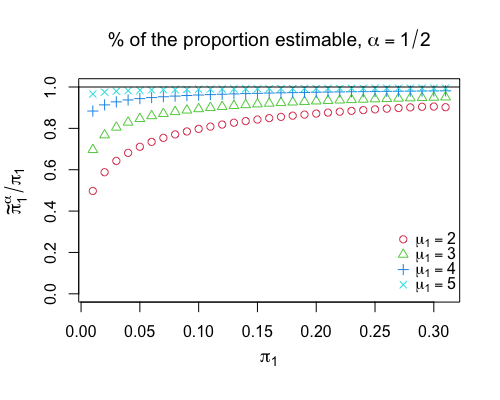}} &
\subfloat{\includegraphics[width=.45\textwidth]{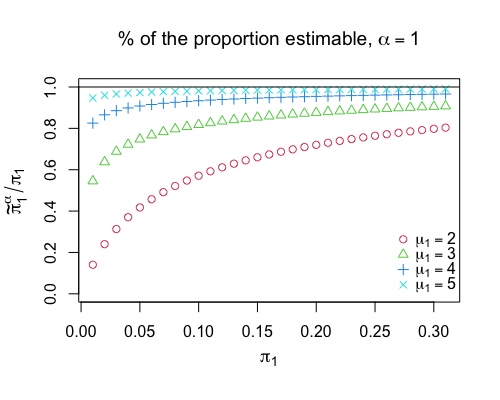}}
\end{tabular}
\caption{ The points on the plots represent the percentage of the false null proportion the DOS-Storey method is able to estimate asymptotically in the Gaussian model $\pi_0N(0,1) + \pi_1N(\mu_1,1)$, $\mu_1 > 0$, for $\alpha = 1/2$ (left) and $\alpha = 1$ (right) and various values of $\pi_1$ ($x$-axis) and different symbols representing the values of $\mu_1$.}
\label{fig:estimable_proportion}
\end{figure}

\section{Additional Simulations}
\label{supp:simul}

\subsection{Dependence on \texorpdfstring{$c_n$}{cn}}

\label{supp:simul_cn}

In this section, we examine the impact of $c_n$ on the DOS-Storey estimates under the Gaussian model. $c_n$ is the proportion of values excluded from the DOS sequence when looking for a maximum (a change-point). We investigate how varying $c_n$ influences the change-point and proportion estimates for different sample sizes $n\in \{100, 1000, 10000, 100000\}$, as well as different values of the nonzero mean $\mu_1$ and the false null proportion $\pi_1$.

The results for the DOS method with $\alpha = 1$ are presented in Figure \ref{fig:cn_dependence}. We included the results for $\alpha = 1$, as they are more variable then those for smaller $\alpha$ values. The left column plots show the mean estimated change-point location as a function of $c$, while the right column plots display the mean estimated proportion as a function of $c$. 
These estimates are averaged over $N = 1000$ repetitions, with each line corresponding to a different sample size, $n$. 
Furthermore, these plots illustrate the convergence of the estimated change-point location and the estimated proportion towards the ideal values defined in Theorem 1 in the main paper, represented by the solid horizontal grey lines.

Notably, in most of the scenarios, the estimates remain stable. The estimates become sensitive when we exclude ``too many'' values which in general happens as $c$ approaches $\pi_1$. Additionally, sensitivity to $c$ is more pronounced when the signal is weaker, as shown in Figure \ref{fig:cn_dependence_weak}. If the nonzero mean is not sufficiently large, the quantile function of the $p$-values becomes smoother, leading to increased variability of the change-point location.

\begin{figure}[H]
\begin{tabular}{cccc}
\subfloat{\includegraphics[width=.45\textwidth]{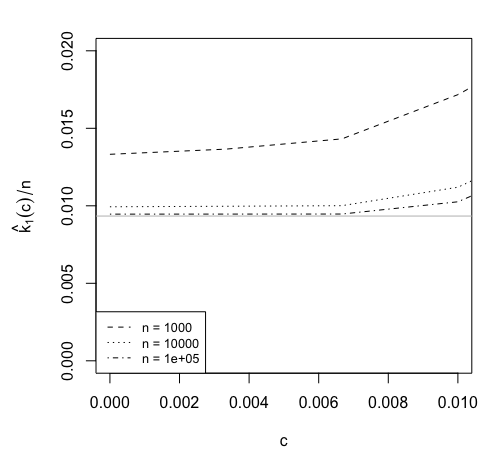}} &
\subfloat{\includegraphics[width=.45\textwidth]{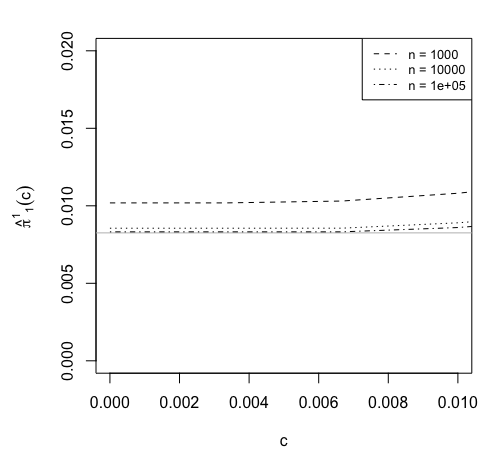}} \\
\subfloat{\includegraphics[width=.45\textwidth]{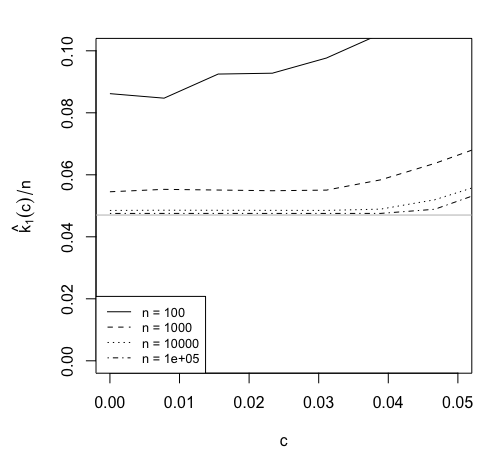}}&
\subfloat{\includegraphics[width=.45\textwidth]{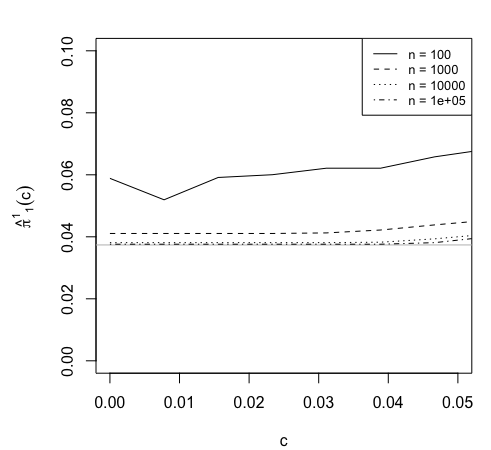}} \\
\subfloat{\includegraphics[width=.45\textwidth]{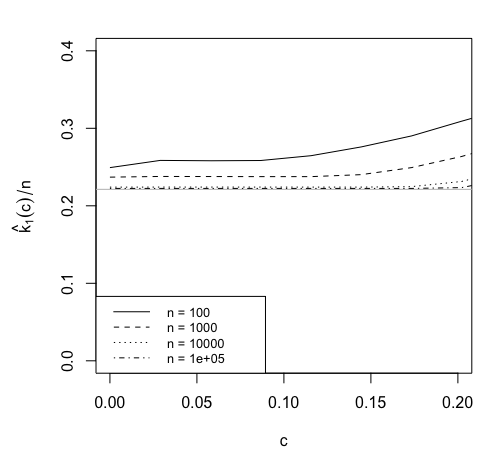}} &
\subfloat{\includegraphics[width=.45\textwidth]{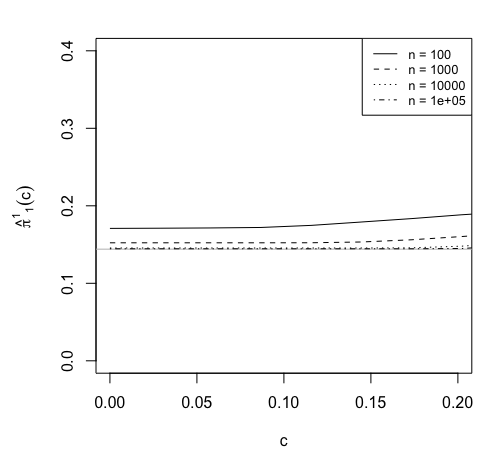}} 
\end{tabular}
\caption{Averaged DOS change-point estimates (left column) and DOS-Storey $\alpha = 1$ proportion estimates $\hat{\pi}_1^1$ (right column) as a function of the proportion of excluded values ($c$) from the DOS sequence. The data is simulated from the Gaussian model. Top row: $\mu_1 = 4, \pi_1 = 0.01$, Middle row: $\mu_1 = 3, \pi_1 = 0.05$, Bottom row: $\mu_1 = 2, \pi_1 = 0.2$.}
\label{fig:cn_dependence}
\end{figure}

\begin{figure}[H]
\begin{tabular}{cc}
\subfloat{\includegraphics[width=.45\textwidth]{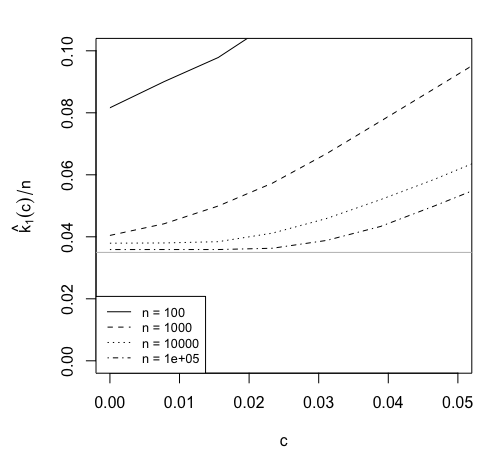}} &
\subfloat{\includegraphics[width=.45\textwidth]{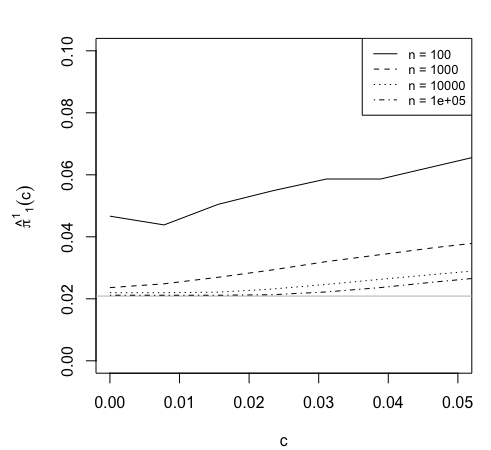}} 
\end{tabular}
\caption{Averaged DOS change-point estimates (left column) and the DOS-Storey $\alpha = 1$ proportion estimates  $\hat{\pi}_1^1$ (right column) as a function of the proportion of excluded values ($c$) from the DOS sequence. The data is simulated from the Gaussian model with $\mu = 2$ and $\pi_1 = 0.05$. This plot shows that for weaker effects the estimates are more affected by the choice of $c$.}
\label{fig:cn_dependence_weak}
\end{figure}

\subsection{Comparing DOS-Storey to Other Change-Point Methods}

\label{supp:simul_cp_methods}

In this section, we compare the performance of the DOS-Storey methods to the other change-point-based estimators in the literature. The simulation settings are as those in Table 1 from the main paper, and the change-point based methods included are by \cite{Hwang2014} (HKWL) and \cite{Turkheimer2001} (TSS). Additionally, as announced in Section 2 of the main paper, we include the method analogous to the proposed DOS-Storey for $\alpha = 1/2$, that searches for a change-point using intervals $(0,i)$ and $(i,n)$, for $i = 1,\dots, n/2$. We call this method CUSUM-Storey (C-S).
The results are shown in Table \ref{tab:cpmethods}. 

The HKWL method operates similarly to the DOS-Storey method. It involves fitting a piecewise linear function with one change-point to the sequence of $p$-values and then using the $p$-value at the identified change-point location as the parameter $\lambda$ for Storey's estimator. However, this estimator tends to produce overly conservative estimates. The DOS-Storey estimates consistently outperform the HKWL estimates. The TSS method yields more precise estimates compared to the HKWL method. Nonetheless, aDOS-Storey outperforms TSS method in all cases in terms of RMSE. C-S method yields reasonable estimates, however in most of the cases it is outperformed by the adaptive $\hat{\pi}_1^a$ estimator.

\begin{table}[]
\centering
\begin{tabular}{lrrrrrr}
     & $\hat{\pi}_1^1$                 & $\hat{\pi}_1^{1/2}$                 & $\hat{\pi}_1^a$                  & HKWL                & TSS & C-S  \\ \hline
\multicolumn{6}{c}{$\mu_1= 3.5 , \pi_1= 0.01 , n_1 = 10$}                                                   \\ \hline
BIAS & -0.50                & 10.30                 & {\ul \textbf{-0.40}}  & -7.70                & -0.30  & 5.20\\
SD   & 3.80                 & 16.10                 & 4.50                  & {\ul \textbf{1.30}}  & 20.90  & 14.30\\
RMSE & {\ul \textbf{3.90}}  & 19.10                 & 4.50                  & 7.80                 & 20.90  & 15.20\\ \hline
\multicolumn{6}{c}{$\mu_1= 3.5 , \pi_1= 0.03, n_1 = 30 $}                                                   \\ \hline
BIAS & -2.90                & 6.10                  & -2.20  & -18.80               & -2.00  & {\ul \textbf{-1.80}}\\
SD   & 5.20                 & 13.10                 & 7.10                  & {\ul \textbf{4.30}}  & 19.10  & 5.90\\
RMSE & {\ul \textbf{5.90}}  & 14.50                 & 7.50                  & 19.30                & 19.20  & 6.20\\ \hline
\multicolumn{6}{c}{$\mu_1= 3 , \pi_1= 0.1 , n_1 = 50 $}                                                     \\ \hline
BIAS & -8.70                & 4.80                  & {\ul \textbf{1.90}}   & -38.20               & -3.10  & -7.30\\
SD   & 8.30                 & 16.90                 & 17.70                 & {\ul \textbf{5.80}}  & 19.50  & 8.70\\
RMSE & 12.10 & 17.60                 & 17.80                 & 38.60                & 19.80  & {\ul \textbf{11.30}}\\ \hline
\multicolumn{6}{c}{$\mu_1= 2 , \pi_1= 0.1 ,  n_1 = 100   $}                                                 \\ \hline
BIAS & -36.70               & -4.70                 & {\ul \textbf{-6.80}}  & -97.30               & -26.20 & -29.70\\
SD   & 19.80                & 23.70                 & 26.10                 & {\ul \textbf{2.60}}  & 20.40  & 18.60\\
RMSE & 41.70                & {\ul \textbf{24.20}}  & 27.00                 & 97.40                & 33.20  & 35.10\\ \hline
\multicolumn{6}{c}{$\mu_1= 3 , \pi_1= 0.1,  n_1 = 100      $}                                               \\ \hline
BIAS & -13.10               & {\ul \textbf{0.80}}   & {\ul \textbf{0.80}}   & -53.60               & -6.90 & -13.70 \\
SD   & 10.70                & 15.80                 & 15.80                 & {\ul \textbf{9.50}}  & 18.90  & 9.10 \\
RMSE & 16.90                & {\ul \textbf{15.80}}  & {\ul \textbf{15.80}}  & 54.40                & 20.10  & 16.50\\ \hline
\multicolumn{6}{c}{$\mu_1= 2 , \pi_1= 0.2 , n_1 = 200   $}                                                  \\ \hline
BIAS & -47.30               & {\ul \textbf{-16.20}} & {\ul \textbf{-16.20}} & -179.40              & -49.00 & -52.40\\
SD   & 26.60                & 23.20                 & 23.20                 & {\ul \textbf{18.80}} & 20.10  & 21.80\\
RMSE & 54.30                & {\ul \textbf{28.30}}  & {\ul \textbf{28.30}}  & 180.30               & 52.90 & 56.70 \\ \hline
\multicolumn{6}{c}{$\mu_1= 3 , \pi_1= 0.2,  n_1 = 200  $}                                                   \\ \hline
BIAS & -20.30               & {\ul \textbf{-2.70}}  & {\ul \textbf{-2.70}}  & -58.60               & -10.70 & -22.00\\
SD   & 12.60                & 17.20                 & 17.20                 & 11.80 & 17.80  & {\ul \textbf{10.70}} \\
RMSE & 23.90                & {\ul \textbf{17.40}}  & {\ul \textbf{17.40}}  & 59.70                & 20.80  & 24.40\\ \hline
\multicolumn{6}{c}{$\mu_1= 3 , \pi_1= 0.3 , n_1 = 300  $}                                                   \\ \hline
BIAS & -23.70               & {\ul \textbf{-6.20}}  & {\ul \textbf{-6.20}}  & -51.50               & -13.90 & -26.50\\
SD   & 13.60                & 15.40                 & 15.40                 & 13.20 & 16.80  & {\ul \textbf{11.80}}\\
RMSE & 27.30                & {\ul \textbf{16.60}}  & {\ul \textbf{16.60}}  & 53.20                & 21.70  & 29.00
\end{tabular}
\caption{Bias, standard deviation and the RMSE of the estimated number of the false null hypotheses ($n\times \hat{\pi}_1$), given the proportion of false null $p$-values $\pi_1$, and the non-zero mean $\mu_1$, for a sample of size $n=1000$, based on $1000$ repetitions. Bold and underlined values correspond to the smallest values in each row.}
\label{tab:cpmethods}
\end{table}

\subsection{Adaptive FDR Control}

\label{supp:simul_fdr}

In this section, we illustrate the performance of the aDOS method for adaptive FDR control. We use the same threshold as in Section 4 of the main paper, $\tau = n^{-1/2}$, and evaluate its performance under dependence and independence and compare it to adaptive procedures using different estimators. 

The FDR-controlling multiple testing procedure proposed by \cite{benjamini1995controlling} (BH) rejects the hypotheses corresponding to the $\hat{k}$ smallest $p$-values, where $\hat{k} = \max \left\{k\ge 1: p_{(k)} \le \alpha \frac{k}{n} \right\}$ and $p_{(1)} \le \dots \le p_{(n)}$. It achieves effective FDR control at level $\pi_0\alpha$. As suggested by \cite{benjamini2000adaptive}, incorporating a proportion estimator into the BH procedure increases its power by adapting to the unknown proportion. This adaptation is achieved by increasing the FDR level parameter from $\alpha$ to $\alpha' = \alpha/\hat{\pi}_0$, which results in a higher number of rejections while maintaining approximate FDR control at level $\alpha$. The simulation study in this section investigates how different proportion estimators affect FDR control and the power of the resulting adaptive procedures. In the context of an adaptive BH procedure, power is defined as the ratio of the number of true discoveries to the total number of false null hypotheses. Note that in this section, we use $\alpha$ to refer to the level of the FDR control and not the parameter of the DOS sequence from the main paper.

Since the DOS-Storey method is unsuitable when the false null proportion is very large, we restrict our analysis to cases where $\pi_0\ge 0.5$ ($\pi_1\le0.5$). We adopt the simulation settings used in \cite{Blanchard2009}.
We consider various values for the mean under the alternative and various $\pi_0$ values. The desired level of the FDR control is set to $\alpha = 0.05$.
Figures \ref{fig:fdr_control_pi0_grid} and \ref{fig:fdr_control_mu_grid} show the simulation results of this analysis for sample size $n=100$, based on $N=1000$ repetitions. For various adaptive BH procedures, we report the FDR and the power relative to the oracle procedure, which is the adaptive BH that uses the unknown true value of $\pi_0$. In Figure \ref{fig:fdr_control_pi0_grid}, we consider fixed nonzero mean $\mu_1=3$ and changing true null proportion values on the $x$-axes. In Figure \ref{fig:fdr_control_mu_grid} we consider fixed $\pi_0 = 0.75$ and changing nonzero mean $\mu_1$ on the $x$-axes.

The methods included in simulations defined in the main paper are ST-MED, ST-1/2, and LSL. ST-alpha corresponds to Storey's proportion estimator with $\lambda = \alpha$, with the level of the FDR control set to $0.05$. Simulations in \cite{Blanchard2009} show that ST-alpha works particularly well in the presence of dependence. In case of independence, \cite{Blanchard2009} single out ST-1/2 as the best proportion estimator for adaptive FDR control.
The approach proposed by \cite{storey2003statistical}, as implemented in the function \code{pi0est} from the R package \code{qvalue}, requires a larger number of $p$-values near 1. This requirement makes the method unsuitable for situations involving smaller sample sizes or dependency.

In the top right plot of Figure \ref{fig:fdr_control_pi0_grid}, we see that under independence, and if the signal is sparse, the aDOS method has a higher power than the ST-1/2 method. However, the FDR is slightly above $\alpha=0.05$ in this small sample case. Similarly, under dependence, aDOS controls the FDR at the level slightly higher than $\alpha=0.05$. The FDR control is better than that of ST-1/2 and LSL, and the power is larger than that of ST-$\alpha$.
Similar conclusions can be drawn from Figure \ref{fig:fdr_control_mu_grid}. The aDOS method behaves stable under dependence and yields a meaningful adaptive BH procedure under both independence and dependence.

We also include simulation results for larger sample sizes $n\in\{500,1000, 5000\}$ under independence, based on $N = 1000$ repetitions. This sample size allows us to include the method by \cite{storey2003statistical} (STS), which is most commonly used in the applied literature. The FDR and relative power of various adaptive BH procedures are shown in Figure \ref{fig:fdr_control_large_sample}. It can be noted that the FDR control of the STS method is slightly above the desired level $\alpha = 0.05$, which is more prominent in smaller sample sizes. On the other hand, for large sample sizes, the aDOS method is more conservative and keeps the FDR controlled below $\alpha$.

\begin{figure}[H]
\begin{tabular}{cccc}
\subfloat{\includegraphics[width=.45\textwidth]{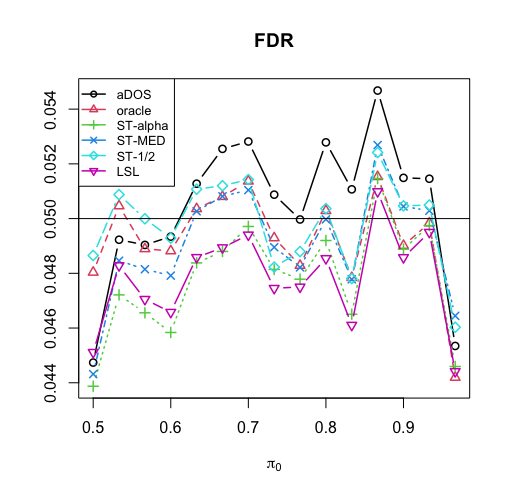}} &
\subfloat{\includegraphics[width=.45\textwidth]{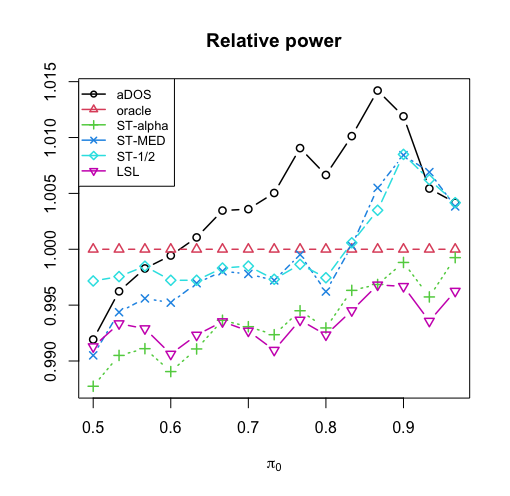}} \\
\subfloat{\includegraphics[width=.45\textwidth]{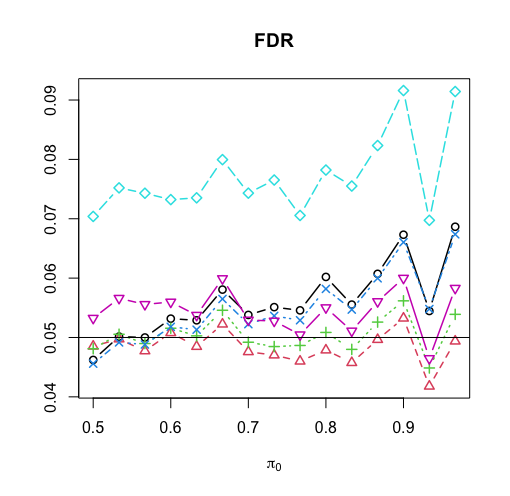}}&
\subfloat{\includegraphics[width=.45\textwidth]{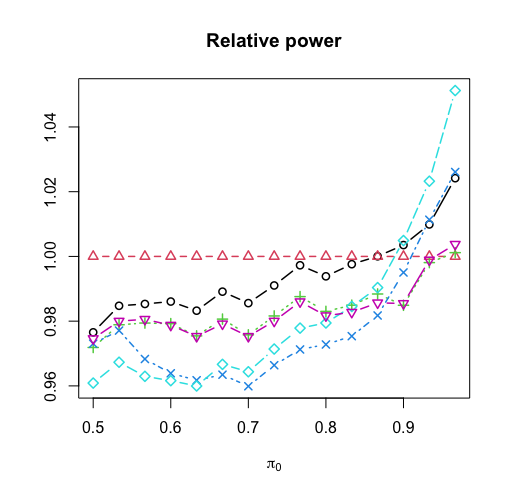}} \\
\subfloat{\includegraphics[width=.45\textwidth]{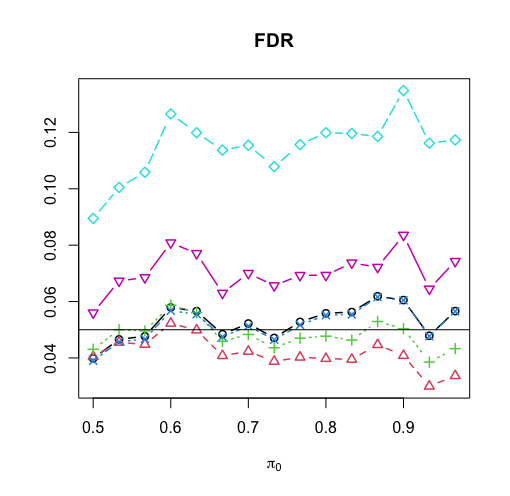}} &
\subfloat{\includegraphics[width=.45\textwidth]{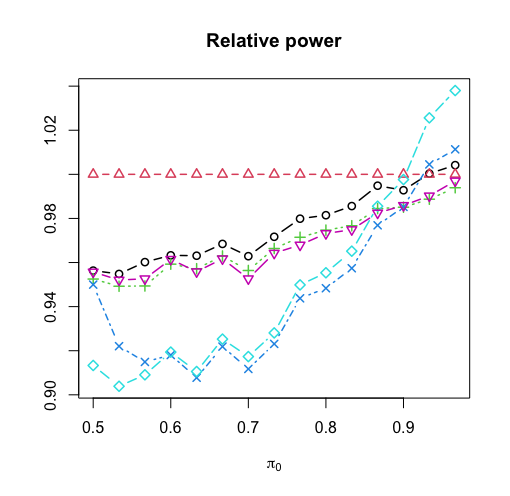}} 
\end{tabular}
\caption{FDR and power relative to oracle as a function of the true null proportion $\pi_0$, with $\mu_1=3$, for various adaptive BH procedures. From top to bottom: correlation coefficient $\rho\in\{0,0.2,0.5\}$.}
\label{fig:fdr_control_pi0_grid}
\end{figure}

\begin{figure}[H]
\begin{tabular}{cccc}
\subfloat{\includegraphics[width=.45\textwidth]{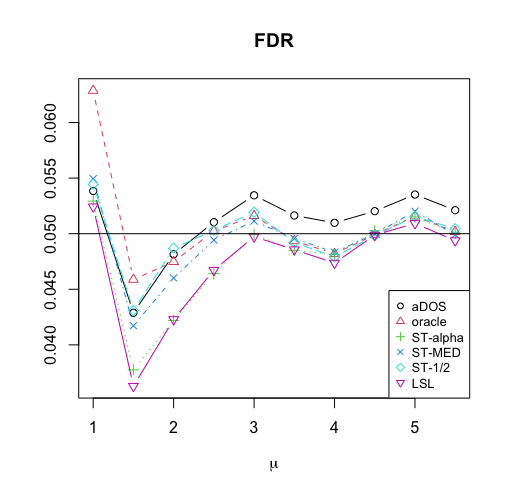}} &
\subfloat{\includegraphics[width=.45\textwidth]{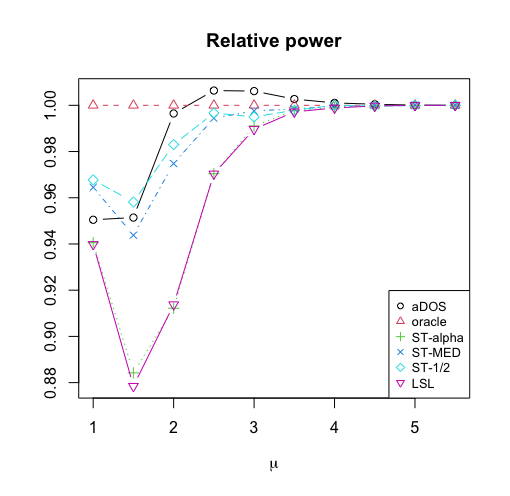}} \\
\subfloat{\includegraphics[width=.45\textwidth]{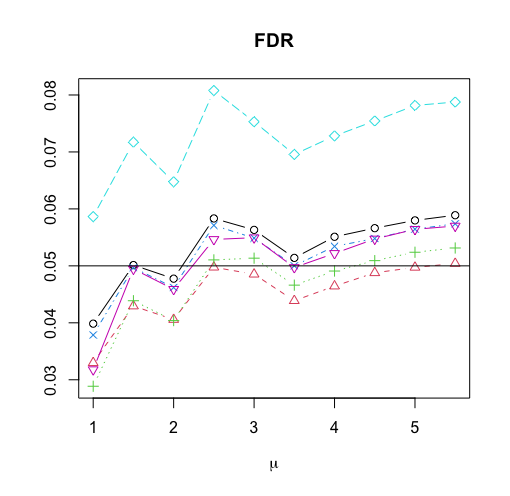}}&
\subfloat{\includegraphics[width=.45\textwidth]{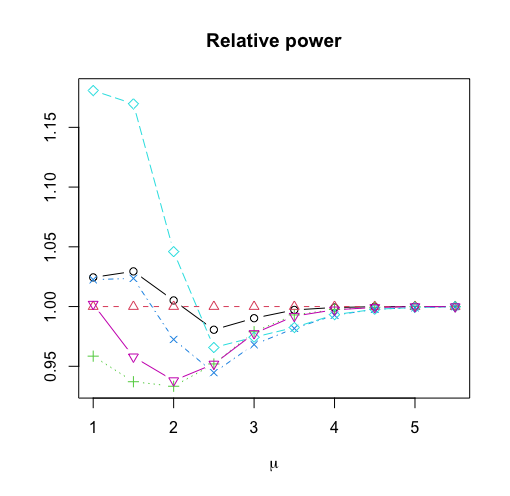}} \\
\subfloat{\includegraphics[width=.45\textwidth]{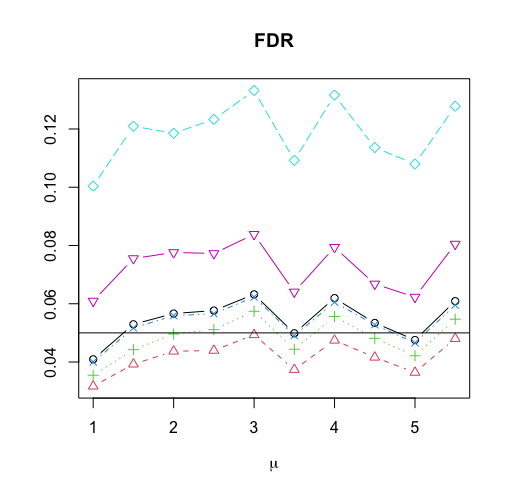}} &
\subfloat{\includegraphics[width=.45\textwidth]{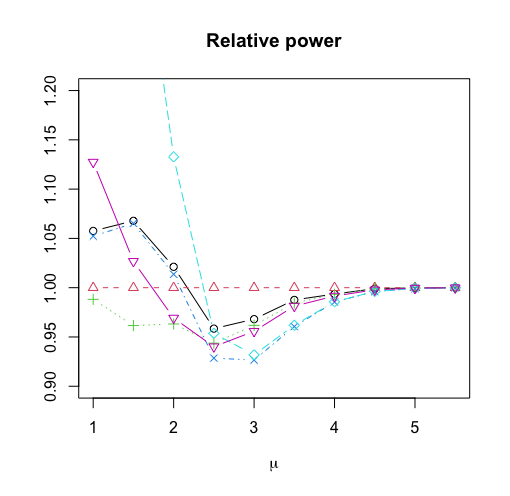}} 
\end{tabular}
\caption{FDR and power relative to oracle as a function of the signal strength under the alternative, where $\pi_0=0.75$, for various adaptive BH procedures. From top to bottom: correlation coefficient $\rho\in\{0,0.2,0.5\}$.}
\label{fig:fdr_control_mu_grid}
\end{figure}

\begin{figure}[H]
\begin{tabular}{cccc}
\subfloat{\includegraphics[width=.45\textwidth]{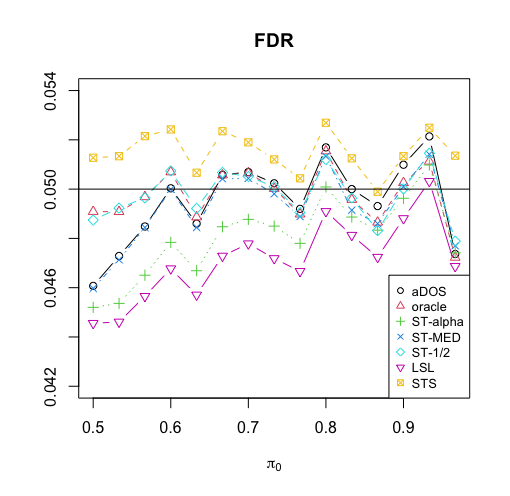}} &
\subfloat{\includegraphics[width=.45\textwidth]{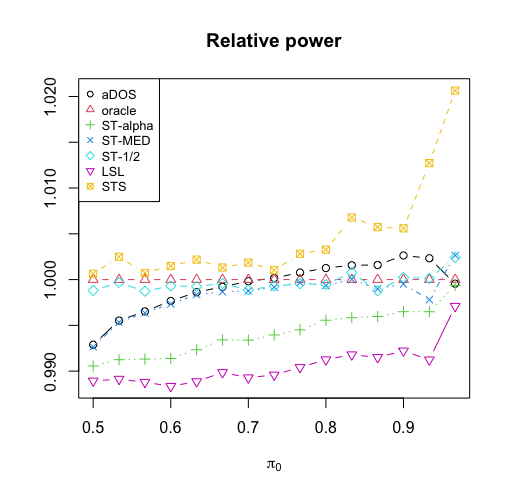}} \\
\subfloat{\includegraphics[width=.45\textwidth]{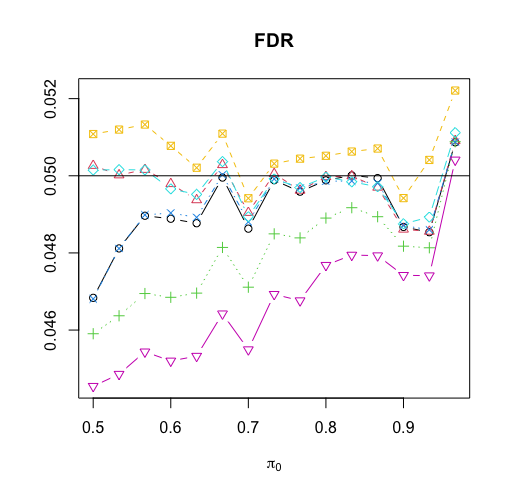}}&
\subfloat{\includegraphics[width=.45\textwidth]{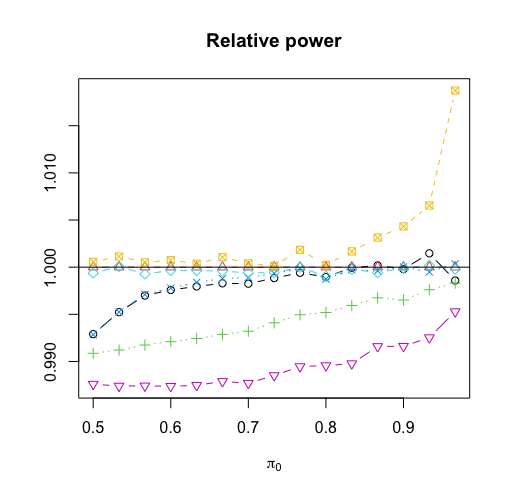}} \\
\subfloat{\includegraphics[width=.45\textwidth]{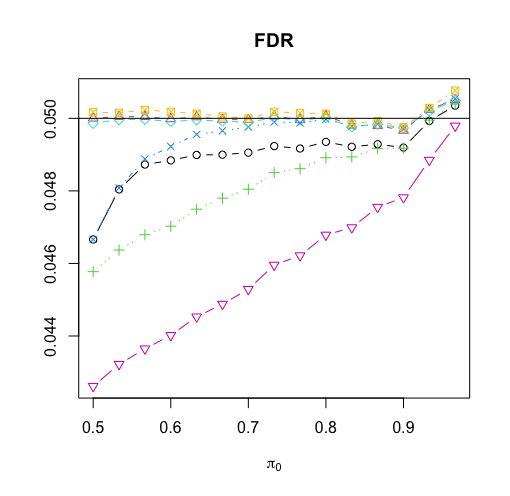}} &
\subfloat{\includegraphics[width=.45\textwidth]{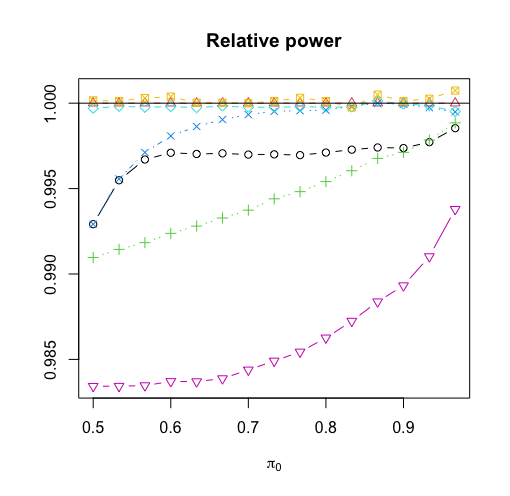}} 
\end{tabular}
\caption{FDR and power relative to oracle as a function of the true null proportion $\pi_0$, with $\mu_1=3$, for various adaptive BH procedures in case of large sample. From top to bottom: the size of the sample is $n\in\{500,1000,10000\}$.}
\label{fig:fdr_control_large_sample}
\end{figure}

\subsection{Superuniform  \texorpdfstring{$p$}{p}-values}
\label{sec:simul_superuniform}

We extend our contribution by considering the proportion estimation in case of superuniform (stochastically larger than uniform) true null $p$-value distributions. This scenario arises when the null hypothesis is misspecified or composite. In such cases, the linearity assumption used by Storey's estimator does not hold, making it unsuitable to use.

We consider $p$-values generated from the composite Gaussian mean testing problem, $H_0: \mu \le 0$ against $H_1: \mu >0$. The distribution under the null is $N(\mu_0,1)$, where $\mu_0\le 0$ and under the alternative $N(\mu_1,1)$ for $\mu_1>0$. The $p$-values are calculated using the least favorable parameter configuration -- when $\mu_0$ is closest to the parameter values under the alternative. This corresponds to $\mu_0 = 0$, so $p_i = 1 - \Phi(T_i)$. This setting produces superuniform $p$-values.

Under the superuniformity assumption, we propose estimating the false null proportion directly from the DOS statistic, using the uDOS estimator defined as:
\begin{equation} \label{eq:uncorr_dos_est}
    \hat{\pi}_{1, \text{uDOS}}^\alpha= \hat{k}_\alpha/n.
\end{equation}
In this approach, the DOS threshold $p_{(\hat{k}_\alpha)}$ acts as a separation threshold for classifying the $p$-values into two groups. The estimated proportion is the proportion of values in the ``small $p$-values group''.

The uDOS proportion estimator for $\alpha = 1$ and $\alpha = 1/2$ is compared to the method proposed in \cite{Hoang2020OnTU} (HD). It involves randomizing $p$-values to enforce uniformity and utilizing the ST-1/2 method on the randomized $p$-values for proportion estimation.

The data is simulated as described at the beginning of this section.
We adopt the mean parameter values used in \cite{Hoang2020OnTU} where the mean under the null is $\mu_0=-0.2r$ and under the alternative $\mu_1=1+0.25r$, for $r\in\{1,...,10\}$.
The simulations are conducted for $n=100$ and based on $N=10000$ repetitions.

In Figure \ref{fig:superuniform_pi1_grid}, we compare the performance of the uDOS proportion estimator (\ref{eq:uncorr_dos_est}) with the estimator proposed by \cite{Hoang2020OnTU} (HD) and Storey's estimator with $\lambda=0.5$ applied to the non-randomized $p$-value sequence.
In terms of the RMSE, the uDOS estimator with $\alpha = 1$ outperforms the HD method overall, although it leads to negatively biased estimates.

\begin{figure}[H]
\begin{tabular}{cccc}
\subfloat{\includegraphics[width=.32\textwidth]{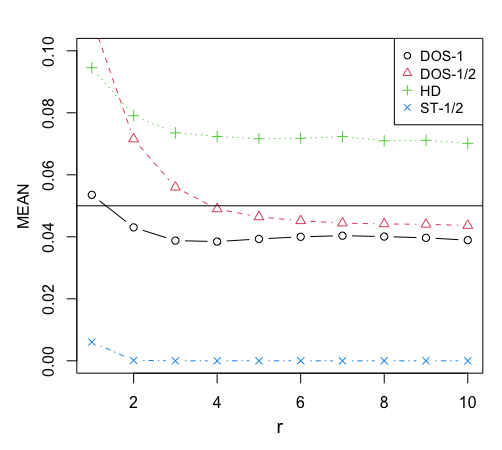}} &
\subfloat{\includegraphics[width=.32\textwidth]{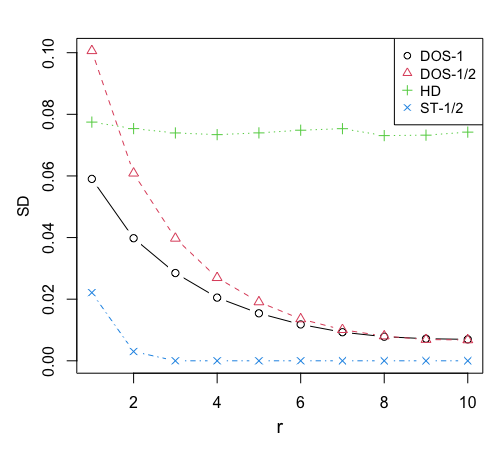}} &
\subfloat{\includegraphics[width=.32\textwidth]{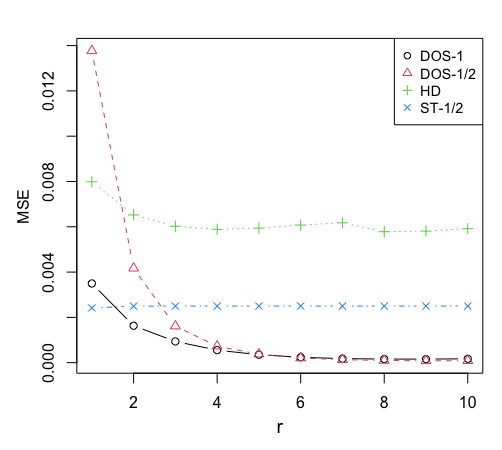}} \\
\subfloat{\includegraphics[width=.32\textwidth]{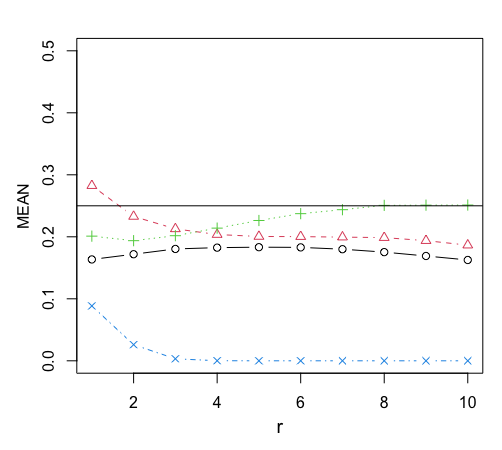}} &
\subfloat{\includegraphics[width=.32\textwidth]{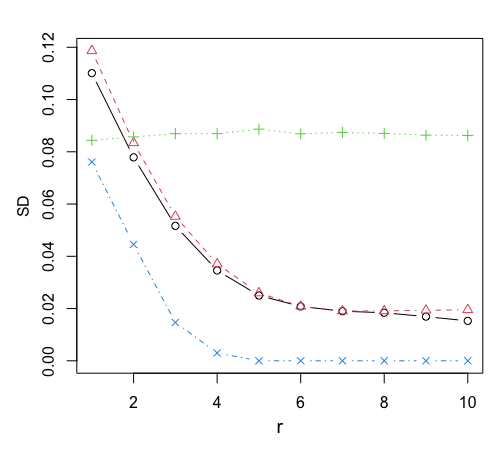}}&
\subfloat{\includegraphics[width=.32\textwidth]{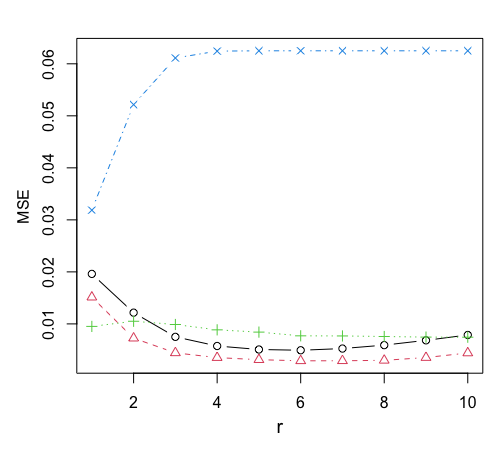}} \\
\end{tabular}
\caption{Mean, standard deviation, and the MSE of different proportion estimators when applied to superuniform true null $p$-values, generated as described in Section \ref{sec:simul_superuniform}. The $x$-axis represents $r$, indicating the distance between the true and false null means. Top: $\pi_1=0.05$. Bottom: $\pi_1 = 0.25$.}
\label{fig:superuniform_pi1_grid}
\end{figure}

\subsection{Dependent \texorpdfstring{$p$}{p}-values}

\label{supp:simul_dependence}

This section provides simulation results for the DOS-Storey and the uDOS estimators under dependence. We now describe the model used to generate the $p$-values, similar to the Gaussian mean testing described in Section 4 but incorporating dependence in the $p$-value sequence.
The test statistics are $T_i=\mu+\eps_i$, where $\mu$ is the mean parameter of the Gaussian distribution we are testing. Under the null hypothesis $\mu=0$ and under the alternative $\mu>0$. Random variables $\eps_i$ are defined as
\begin{equation} \label{eq:define_corr_model}
    \eps_i= \sqrt{\rho}U_i + \sqrt{1-\rho}Z_i.
\end{equation}
where $U_i$ and $Z_i$ are iid $N(0,1)$ random variables and $\rho\in[0,1]$ is a correlation parameter.
For $\rho>0$ this introduces positive correlation in the test statistics, $\text{corr}(T_i,T_j) = \rho$, for $i\neq j$. 
The $p$-values are calculated as $p_i = 1- \Phi(T_i)$. The simulation results are based on the sample size $n = 100$ and various values for the mean under the alternative and various $\pi_0$ values. The number of repetitions is $N = 1000$.

The simulation results for $\rho = 0.2$ are shown in Table \ref{tab:dependence_rho02} and for $\rho = 0.5$ in Table \ref{tab:dependence_rho05}. Both tables show that DOS-Storey methods perform satisfactory compared to other methods in the presence of dependence. While the LSL estimator shows the overall best performance in Table \ref{tab:dependence_rho02}, $\hat{\pi}_1^1$ and $\hat{\pi}_1^{1/2}$ perform the best in  Table \ref{tab:dependence_rho05}.

We also include simulations under dependence for superuniform $p$-values. 
We adopt the simulation setting used in \cite{Hoang2020OnTU}. The mean under the null is $\mu_0=-0.2r$ and under the alternative $\mu_1=1+0.25r$ for $r\in\{1,\dots,10\}$. The correlation is introduced as in (\ref{eq:define_corr_model}) and the $p$-values are calculated as $p_i = 1-\Phi(T_i)$.  All the simulations are conducted for $n=100$.
In Figure \ref{fig:superuniform_rho_grid}, we illustrate the performance of the uDOS proportion estimator under dependence alongside the estimator proposed by \cite{Hoang2020OnTU} (HD). We show the estimators' average mean, standard deviation, and RMSE based on $N = 10000$ repetitions.
The results are similar to the independent case discussed in the main paper. Both DOS-Storey estimates have negative bias but uniformly smaller MSE than the HD estimates.

Note that although the RMSE of the DOS estimator is smaller than that of the HD, the unbiasedness of the HD estimator plays a more significant role in adaptive FDR control, resulting in more accurate FDR control.

\begin{table}[]
\begin{tabular}{cccccccccc}
     & $\hat{\pi}_1^1$                 & $\hat{\pi}_1^{1/2}$                & ST-1/2               & ST-MED & JD    & LLF   & LSL                  & MGF                 & MR                   \\ \hline
\multicolumn{10}{c}{$ \mu = 3 , \pi_1 = 0.05 $}                                                                                                                         \\ \hline
BIAS & 4.60                 & 7.90                 & 12.90                & 7.30   & 15.10 & 24.50 & {\ul \textbf{-0.60}} & 8.90               & 4.30                 \\
SD   & 12.60                & 13.00                & 22.30                & 13.70  & 22.70 & 32.90 & {\ul \textbf{8.20}}  & 17.40             & 14.80                \\
RMSE & 13.40                & 15.30                & 25.80                & 15.60  & 27.30 & 41.00 & {\ul \textbf{8.20}}  & 19.60             & 15.40                \\ \hline
\multicolumn{10}{c}{$ \mu = 2 , \pi_1 = 0.1 $}                                                                                                                          \\ \hline
BIAS & 1.50                 & 3.90                 & 8.60                 & 3.00   & 11.40 & 20.70 & -4.80                & 4.50                & {\ul \textbf{-0.80}} \\
SD   & 13.40                & 13.20                & 22.40                & 14.00  & 23.30 & 32.70 & {\ul \textbf{9.50}}  & 16.70               & 14.70                \\
RMSE & 13.40                & 13.70                & 24.00                & 14.30  & 26.00 & 38.70 & {\ul \textbf{10.60}} & 17.40               & 14.80                \\ \hline
\multicolumn{10}{c}{$ \mu = 3 , \pi_1 = 0.1 $}                                                                                                                          \\ \hline
BIAS & 2.40                 & 4.10                 & 7.60                 & 2.20   & 9.30  & 17.50 & -2.40                & 3.70                & {\ul \textbf{0.10}}  \\
SD   & 11.90                & 12.20                & 23.10                & 14.10  & 22.90 & 30.20 & {\ul \textbf{10.60}} & 17.00               & 14.40                \\
RMSE & 12.10 & 12.90                & 24.40                & 14.20  & 24.70 & 34.90 & {\ul \textbf{10.80}}                & 17.40               & 14.40                \\ \hline
\multicolumn{10}{c}{$\mu = 2 , \pi_1 = 0.2 $}                                                                                                                           \\ \hline
BIAS & -2.30                & -0.70                & 5.00                 & -2.10  & 6.70  & 15.90 & -9.10                & {\ul \textbf{0.20}} & -6.50                \\
SD   & 14.20                & 13.90                & 24.80                & 15.20  & 25.00 & 32.20 & {\ul \textbf{12.60}} & 18.80               & 16.20                \\
RMSE & 14.40                & {\ul \textbf{13.90}} & 25.30                & 15.40  & 25.90 & 35.90 & 15.60                & 18.80               & 17.40                \\ \hline
\multicolumn{10}{c}{$\mu = 3 , \pi_1 = 0.2 $}                                                                                                                           \\ \hline
BIAS & {\ul \textbf{0.60}}  & 2.20                 & 4.50                 & -1.40  & 6.30  & 17.30 & -4.30                & 1.60                & -3.50                \\
SD   & 10.40                & 10.40                & 22.50                & 14.30  & 22.90 & 27.80 & {\ul \textbf{9.10}}  & 17.40               & 12.40                \\
RMSE & 10.40                & 10.60                & 22.90                & 14.40  & 23.80 & 32.70 & {\ul \textbf{10.00}} & 17.40               & 12.90                \\ \hline
\multicolumn{10}{c}{$\mu = 3 , \pi_1 = 0.3 $}                                                                                                                           \\ \hline
BIAS & -2.10                & -0.70                & {\ul \textbf{-0.40}} & -6.00  & 1.80  & 13.30 & -5.20                & -1.70               & -6.50                \\
SD   & 10.00                & {\ul \textbf{9.70}}  & 24.90                & 15.50  & 25.80 & 26.80 & 12.40                & 19.70               & 13.10                \\
RMSE & 10.20                & {\ul \textbf{9.80}}  & 24.90                & 16.60  & 25.90 & 29.90 & 13.40                & 19.70               & 14.60                \\ \hline
\end{tabular}
\caption{Simulations under dependence, where $p$-values follow the model described in Section \ref{supp:simul_dependence} with $\rho = 0.2$. Bias, standard deviation, and the RMSE of the estimated number of the false null hypotheses ($n\times \hat{\pi}_1$), given the proportion of false null $p$-values $\pi_1$, and the non-zero mean $\mu_1$, for a sample of size $n=100$, based on $200$ repetitions. Bold and underlined values correspond to the smallest values in each row.}
\label{tab:dependence_rho02}
\end{table}

\begin{table}[]
\begin{tabular}{cccccccccc}
     & $\hat{\pi}_1^1$                & $\hat{\pi}_1^{1/2}$               & ST-1/2             & ST-MED              & JD   & LLF  & LSL                & MGF     & MR                 \\ \hline
\multicolumn{10}{c}{$\mu = 3 , \pi_1 = 0.05 $}                                                                                                                \\ \hline
BIAS & 5.1                 & 9.1                 & 23.4               & 10.0                & 26.7 & 38.0 & {\ul \textbf{4.7}} & 21.0   & 15.6               \\
SD   & {\ul \textbf{16.3}} & 17.0                & 34.0               & 16.8                & 34.4 & 44.0 & 25.5               & 30.6  & 24.3               \\
RMSE & {\ul \textbf{17.0}} & 19.3                & 41.2               & 19.5                & 43.5 & 58.2 & 25.9               & 37.1  & 28.9               \\ \hline
\multicolumn{10}{c}{$\mu = 2 , \pi_1 = 0.1 $}                                                                                                                \\ \hline
BIAS & 1.2                 & 3.8                 & 18.0               & 4.9                 & 21.8 & 34.6 & {\ul \textbf{0.1}} & 12.1   & 10.4               \\
SD   & 17.1                & 17.0                & 33.2               & {\ul \textbf{16.6}} & 34.2 & 44.0 & 24.6               & 26.7  & 24.0               \\
RMSE & {\ul \textbf{17.1}} & 17.4                & 37.8               & 17.3                & 40.6 & 56.0 & 24.6               & 29.3  & 26.2               \\ \hline
\multicolumn{10}{c}{$\mu = 3 , \pi_1 = 0.1 $}                                                                                                                \\ \hline
BIAS & {\ul \textbf{1.0}}  & 3.1                 & 14.6               & 3.5                 & 18.1 & 27.5 & {\ul \textbf{1.0}} & 11.9  & 7.7                \\
SD   & {\ul \textbf{15.7}} & 16.3                & 32.9               & 16.7                & 33.4 & 42.6 & 25.0               & 28.8 & 23.8               \\
RMSE & {\ul \textbf{15.7}} & 16.6                & 36.0               & 17.0                & 38.0 & 50.8 & 25.0               & 31.2 & 25.0               \\ \hline
\multicolumn{10}{c}{$\mu = 2 , \pi_1 = 0.2 $}                                                                                                                \\ \hline
BIAS & -4.8                & -2.3                & 12.1               & {\ul \textbf{-2.0}} & 16.0 & 28.1 & -5.3               & 7.5  & {\ul \textbf{2.0}} \\
SD   & 18.2                & {\ul \textbf{17.7}} & 34.8               & {\ul \textbf{17.7}} & 35.1 & 43.4 & 27.4               & 29.2 & 25.1               \\
RMSE & 18.8                & 17.9                & 36.8               & {\ul \textbf{17.8}} & 38.5 & 51.7 & 27.9               & 30.1 & 25.2               \\ \hline
\multicolumn{10}{c}{$\mu = 3 ,\pi_1 = 0.2$}                                                                                                                 \\ \hline
BIAS & -2.1                & {\ul \textbf{-0.6}} & 11.1               & -1.9                & 15.9 & 28.3 & -4.5               & 7.1   & 2.6                \\
SD   & {\ul \textbf{15.1}} & 15.3                & 32.6               & 17.0                & 33.6 & 41.3 & 21.3               & 26.7  & 23.4               \\
RMSE & {\ul \textbf{15.3}} & {\ul \textbf{15.3}} & 34.4               & 17.1                & 37.2 & 50.0 & 21.8               & 27.6  & 23.5               \\ \hline
\multicolumn{10}{c}{$\mu = 3 , \pi_1 = 0.3 $}                                                                                                                \\ \hline
BIAS & -6.8                & -5.7                & 4.2 & -9.2                & 8.3  & 19.6 & -7.1               & {\ul \textbf{0.0}}    & -4.6               \\
SD   & 15.2                & {\ul \textbf{15.1}} & 34.6               & 18.4                & 35.4 & 39.8 & 24.6               & 27.9   & 24.4               \\
RMSE & 16.7                & {\ul \textbf{16.2}} & 34.9               & 20.6                & 36.3 & 44.4 & 25.6               & 27.9   & 24.8               \\ \hline
\end{tabular}
\caption{Simulations under dependence, where $p$-values follow the model described in Section \ref{supp:simul_dependence} with $\rho = 0.5$. Bias, standard deviation, and the RMSE of the estimated number of the false null hypotheses ($n\times \hat{\pi}_1$), given the proportion of false null $p$-values $\pi_1$, and the non-zero mean $\mu_1$, for a sample of size $n=100$, based on $200$ repetitions. Bold and underlined values correspond to the smallest values in each row.}
\label{tab:dependence_rho05}
\end{table}

\begin{figure}[H]
\begin{tabular}{cccc}
\subfloat{\includegraphics[width=.32\textwidth]{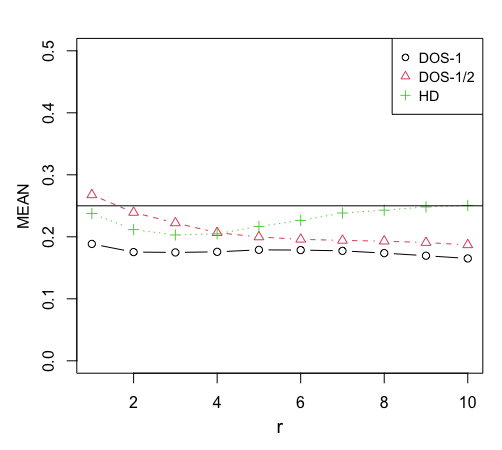}} &
\subfloat{\includegraphics[width=.32\textwidth]{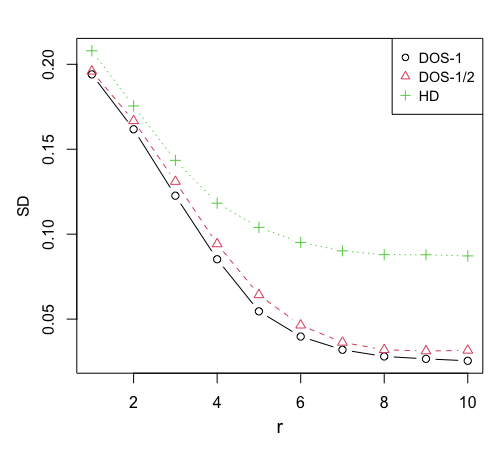}} &
\subfloat{\includegraphics[width=.32\textwidth]{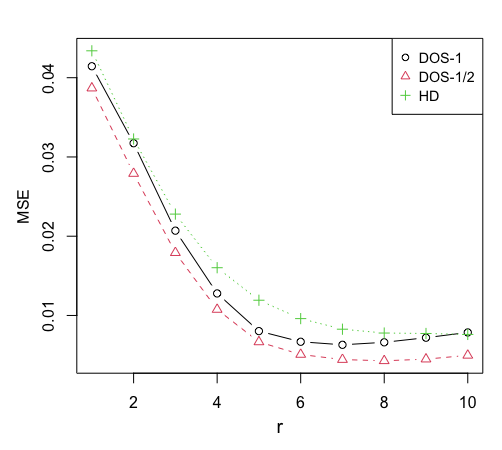}} \\
\subfloat{\includegraphics[width=.32\textwidth]{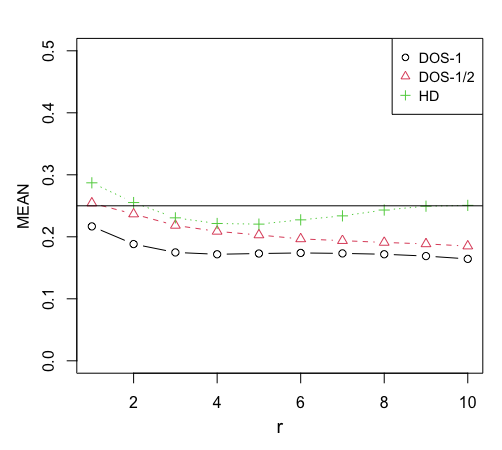}}&
\subfloat{\includegraphics[width=.32\textwidth]{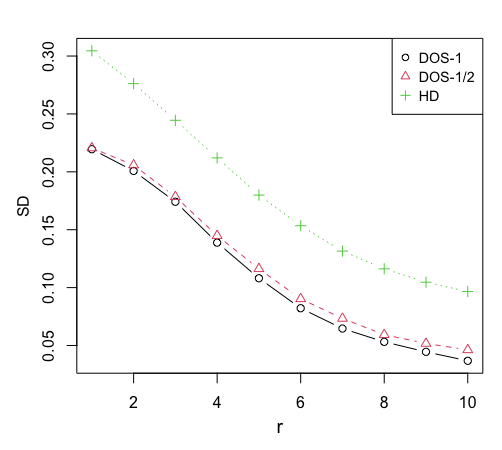}} &
\subfloat{\includegraphics[width=.32\textwidth]{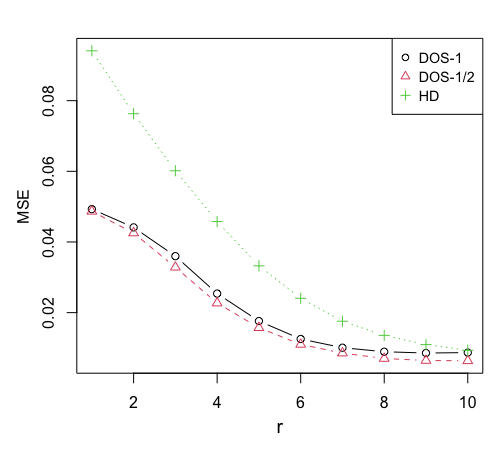}} \\
\subfloat{\includegraphics[width=.32\textwidth]{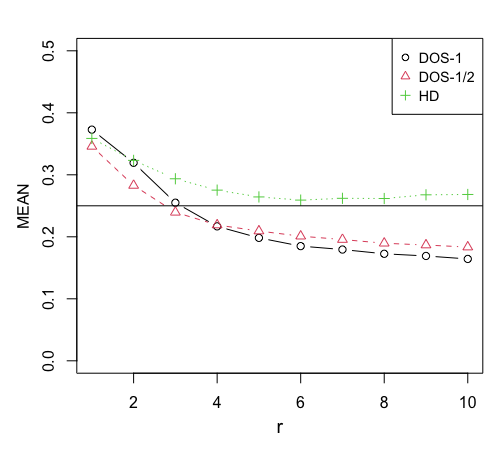}} &
\subfloat{\includegraphics[width=.32\textwidth]{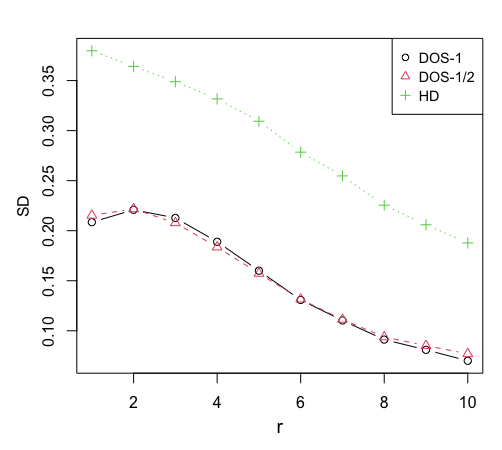}}&
\subfloat{\includegraphics[width=.32\textwidth]{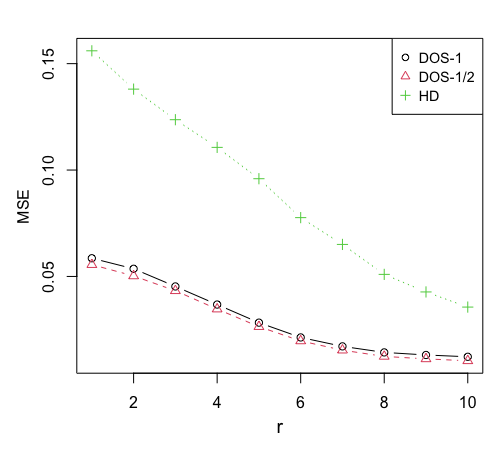}} \\
\end{tabular}
\caption{ Mean, standard deviation, and the MSE of different proportion estimators when applied to superuniform true null $p$-values, generated as described in Section \ref{supp:simul_dependence}. The $x$-axis represents $r$, indicating the distance between the true and false null means. From top to bottom: correlation coefficient $\rho\in\{0.25,0.5,0.75\}$.}
\label{fig:superuniform_rho_grid}
\end{figure}

\section{Interpretations}
\label{supp:interpr}
In this section we analyse how the DOS statistic relates to the methods available in the literature for estimating the knee/elbow point in a plot, by providing a short review of the papers on that topic. We proceed to give a short interpretation of the objective function $h$ defined in Eq. (6) of the main document.

\subsection{Knee/Elbow Estimation}
\label{supp:knee_elbow}

The methods used for estimating the knee/elbow in a plot are often heuristic, and the mathematical definition of this point is sometimes avoided. While some methods concentrate on estimating the point with the largest second derivative, the most natural approach to defining the knee/elbow is as the point where the curvature is maximized.
The DOS statistic shares similarities with existing knee/elbow estimation methods - more precisely for elbow estimation, as the function of interest is convex. However, from this perspective, our objective function $h_F^\alpha$ presents a distinct definition for the elbow, better suited for the discrete nature of the data.

Detecting the knee/elbow in a graph is a problem of interest in model selection problems, such as estimating the number of factors in factor analysis or determining the optimal number of clusters. In their work, \cite{Salvador2004} review several approaches and propose their method. Standard methods include finding the largest difference or the largest ratio between two consecutive points.
The L-method proposed by \cite{Salvador2004} consists of fitting two straight lines on each side of the candidate knee point and selecting the point with the smallest mean squared error as the estimated knee. In \cite{Antunes2018} and \cite{Antunes2018b}, the authors propose a thresholding-based method to identify the sharp angle in the plot, aiming to find the optimal transition point between the high and low values of the first derivative. Another widely used knee detection algorithm is proposed in \cite{Satopaa2011}.

The DOS statistic can be seen as a method for estimating an elbow in a $p$-value plot.
It can also be compared to the method for finding the sharpest angle in the graph of the quantile function. Simple analysis shows that the ideal change-point location comes after the point where the sharpest angle is between points $(0,i/n)$ and $(i/n,2i/n)$ on the graph. 
In that sense, the DOS statistic may be better seen as a method for estimating the point where the curvature drops significantly rather than the method for finding the point of the maximum curvature.

\subsection{Scanning for the Largest Difference}
\label{sec:DOS/ideal_behaviour_curvature}

We again take a look at the objective function $h$ defined in Eq. (6) of the main document as:
\begin{equation}
    h_F^\alpha(t)=\frac{F^{-1}(2t)-2F^{-1}(t)}{t^\alpha}, \quad t\in (0,1/2).
\end{equation}
Define
\begin{align} \label{eq:family_alpha_H_functions}
   H^\alpha(t,a)&\defeq \frac{F^{-1}(t+a)-F^{-1}(t)}{a^\alpha}-\frac{F^{-1}(t)-F^{-1}(t-a)}{a^\alpha},
\end{align}
for $a\in (0,t]$ and $t\in [0,1]$.
For $a=t$, it holds that $H^\alpha (t,t)= h_F^\alpha(t)$.

Under the assumption of the existence and continuity of the first two derivatives of $(F^{-1})(t)$ on $t\in[0,1]$, and considering that $(F^{-1})'$ is an increasing function, it follows that  $H^\alpha(t,a)$ is increasing in $a$ for any fixed $t$. Let $\tilde{t}_\alpha = \argmax_t h_F^\alpha(t)$. If $\tilde{t}_\alpha\le 1/2$, it holds that
\begin{equation*}   (\tilde{t}_\alpha,\tilde{t}_\alpha)=\argmax_{(t,a)} H^\alpha (t,a) = \argmax_{(t,t)} H^\alpha (t,t),
\end{equation*}
suggesting that the DOS statistic essentially estimates the point of the largest scaled difference on symmetric intervals of arbitrary length in the quantile function. We note that this interpretation is possible under the concavity assumption and assumption (A1). Scanning through all possible window sizes $a$ is unnecessary and would introduce additional noise to the estimator, however, this perspective provides an interpretation that resembles the scan statistic commonly used in change-point analysis.

\bibliographystyle{jabes}
\bibliography{supp_technometrics}